\documentclass[preprint,12pt,authoryear]{elsarticle}
\usepackage{booktabs}
\usepackage{amsfonts}
\usepackage{amsthm}
\usepackage{amsmath}
\usepackage{rotating}
\usepackage{txfonts}
\usepackage{multirow}
\usepackage{amsbsy}
\usepackage{psfrag}
\usepackage{epsf}
\usepackage{graphicx}
\usepackage{makeidx}
\usepackage{fancybox}
\usepackage{CJK}
\usepackage{color}
\usepackage{bm}
\usepackage{threeparttable}
\usepackage{pythonhighlight}
\usepackage{titlesec}
\usepackage{caption}
\usepackage{fancyhdr}
\usepackage{placeins}
\usepackage[colorlinks,citecolor=blue,urlcolor=blue]{hyperref}
\usepackage{graphicx}
\usepackage{longtable}
\usepackage{amssymb}
\usepackage{hyperref}
\usepackage{url}
\usepackage{amsthm,color}
\usepackage{array}
\usepackage{amsthm,color}
\usepackage{mathrsfs}
\usepackage{amsmath}
\usepackage{amssymb}
\usepackage{graphicx}
\usepackage{epsfig}
\usepackage{lscape}
\usepackage{makecell}
\newcolumntype{P}[1]{>{\centering\arraybackslash}p{#1}}
\newcolumntype{M}[1]{>{\centering\arraybackslash}m{#1}}
\usepackage[ruled,linesnumbered]{algorithm2e}
\SetKwProg{Fn}{Function}{:}{end}
\SetKwFunction{Pv}{pvalue}
\SetKwFunction{Bo}{bound}
\usepackage{listings}
\usepackage{xcolor}

\definecolor{dkgreen}{rgb}{0,0.45,0}
\definecolor{mauve}{rgb}{0.58,0,0.82}

\newtheorem{theorem}{Theorem}

\newtheorem{corollary}{Corollary}
\newtheorem{exam}{Example}
\newtheorem{prop}{Proposition}
\newtheorem{rem}{Remark}

\begin{document}


\begin{center}
\Large {\bf Causal inference of Plackett-Burman designs in applications}

\end{center}
\bigbreak

\begin{center}
{\normalsize Shuchen Chang, Zhi-ming Li\footnote{
Address correspondence to Zhiming Li, College of Mathematics and System Science, Xinjiang
University, Urumqi,  China;
E-mail: zmli@xju.edu.cn}}
\end{center}

\begin{center}
    {\normalsize College of Mathematics and System Science, Xinjiang University, Urumqi, China}
\end{center}
\bigbreak

\small\begin{quote}{\it
{\bf Abstract} 
Driven by four applications of Plackett-Burman (PB) designs, this paper proposes a causal inference framework based on potential outcomes. First, we define the causal effects of the PB designs under finite populations. The Neymanian estimator of causal effects is then obtained, including the estimated variance and covariance.  Furthermore, we conduct a sharp null-hypothesis test and construct the Fisherian interval using an algorithm. Finally, the proposed methods are illustrated through these applications.
}

\end{quote}\normalsize

\vskip 4mm
\section{Introduction}
The Plackett-Burman (PB) design is a popular and economical class of screening designs, proposed by \cite{plackett1946design}, to improve quality control and support informed decision-making. Compared with two-level full designs, such as 32 runs for 5 factors and 64 runs for 6 factors, PB designs only require $N$ runs to complete an experiment with $N-1$ factors, where $N$ is a multiple of four. {\cite{barrentine1996illustration} divided Placket-Burman designs into two types: geometric and non-geometric. The former is fractional factorial designs yielding $2^{(N-1)-(N-1-q)}$ for run size $N=8,16, 32, 64, \ldots$, and so forth, where $q$ is the number of independent columns (factors). The latter can be generated by the cyclic permutation of the first row, and includes run size $N = 4k$ but not a power of 2 (e.g., $n=12, 20, 24,\ldots$). In the geometric PB design, the relationships among effects are either orthogonal or fully aliased. However, non-geometric PB designs are non-regular, and the effects may be partially aliased.  \cite{wu2011experiments} pointed out that PB design is typically applied at the initial stage of experimental design to rapidly identify the key factors that significantly affect the response variables from a relatively large pool of candidate factors, thereby providing a basis for subsequent experiments.

So far, PB design has been widely applied across various fields, including environment, biomedicine, and engineering. Next, we present four examples to show the applications of PB design. 

\begin{exam}\label{exam1} (\cite{filgueiras2021plackett})
In marine environments, microplastics are stored in the form of sediment. Due to the lack of a standardized extraction, a PB design with 8 runs was used to screen the factors affecting microplastic extraction, with two response variables: fibres and fragments. Seven independent variables were included in the PB design, namely sediment, extraction number, sediment amount, sedimentation time, density separation solution volume, agitation time, and blank volume, denoted by factors 1-7, respectively. Each run was repeated 5 times, and the average concentration $\pm$ standard deviation was given in each run. 
\end{exam}

 In Example \ref{exam1}, the PB design is a $2^{7-4}$ factorial design, and then it is geometric.  The significance of all factors was assessed by checking whether the effect size exceeded twice the average standard deviation. The results showed that none of the factors significantly influenced fibres, while factors 3 and 6 significantly influenced fragments.
 
\begin{exam}\label{exam2} (\cite{sahu2017screening}) Gedunin is a compound with anticancer and anti-inflammatory properties, while it has poor water solubility and low stability. Although liposomes can enhance the solubility and stability of Gedunin, their preparation process is complicated. Thus, a PB design with 11 factors and 12 runs was used to identify which factors significantly affect four liposome properties: vesicle size (VS), zeta potential (ZP), entrapment efficiency (EE), and drug-loading capacity (LC).  For simplicity, the eleven factors are labeled as Factors 1-11, corresponding to drug concentration, lipid concentration, cholesterol/lecithin ratio, chloroform/methanol ratio, volume of organic solvent, volume of aqueous vehicle, volume of aqueous vehicle, rehydrated solution pH, Water bath temperature, rotary evaporator rotation time, rotary evaporator rotation speed, and ultrasonication time, respectively. Each run was repeated 3 times, and the average $\pm$ standard deviation was given for each run.
\end{exam}

The PB design in Example \ref{exam2} is non-geometric. Based on the complete PB design, a linear model was established for the response variables and all main effects. Through analysis of variance (ANOVA), the significant factors affecting VS, ZP, EE, and LC are respectively factors 1, 2, 3, and 11; factors  2, 3, and 11;  factors 1, 2, 3, 4, 8, 9, and 11; and factors 1 and 2. Combining the trends observed in the 3D response surface plot, factors 1, 2, 3, 9, and 11 were identified as the five candidate factors for further design optimization.

\begin{exam}\label{exam3} (\cite{dhat2017risk}) 
Satranidazole (AST) is a high-dose, poorly water-soluble antiprotozoal drug used to treat intestinal amoebiasis. The low solubility and poor bioavailability make formulation development challenging, and the nanoprecipitation technique can improve both the drug's solubility and targeting ability. To optimize the formulation, a 12-run PB design with 8 factors was used to screen for the significant factors affecting four response variables: mean particle size (MPS), ZP, EE, and dissolution efficiency at 30 min (DE$_{30}$).  The eight independent variables were the amount of SAT in the organic phase, type of polymer in the organic phase, type of stabilizer in the aqueous phase, concentration of stabilizer, amount of polymer in the organic phase, volume of aqueous phase, stirring time, and stirring speed, which were coded as factors 1-8, respectively. Each run was repeated 3 times, and the mean $\pm$ standard deviation for each run.
\end{exam}

Base on the incomplete PB design, multiple regression analysis, together with ANOVA,  showed that the significant factors respectively affected the MPS, ZP,  EE, and DE$_{30}$ variable in Example \ref{exam3}:  factors 1, 3, 4, and 6; the only significant factor 3; factors 2, 3, and 5;  and factors 1, 2, and 5. Based on the overall PB screening results, factors 4, 5, and 6 were identified as the critical factors for subsequent optimization: PVA concentration in the aqueous phase, external aqueous phase volume, and ES100 amount. Then, using these factors constructed a $2^3$ full factorial design for further study, re-coded as 1, 2, and 3. The results showed that factors 1 and 3 significantly affected MPS and EE, whereas neither the main effects nor the interaction terms affected the polydispersity index (PDI).

\begin{exam}\label{exam4} (\cite{hardianto2023identification} )
In recent years, electrochemical biosensors have advanced rapidly, driving an increasing demand for products with low cost, high sensitivity, and other desirable characteristics. To address this, a 12-run PB design with 10 factors was applied to screen the key factors influencing the current response $y$ of electrochemical biosensors. The ten factors, namely AuNP volume, DPV potential range, ssDNA probe concentration, immobilization incubation time of the ssDNA probe, AgNO$_3$ concentration, AgNO$_3$ settling time, NaBH$_4$ concentration, NaBH$_4$ incubation time, Ag electrode potential, and duration of Ag electrodeposition, were denoted as factors 1–10. Each run was repeated 2 times, and the observations of each time were provided. 
\end{exam}

Similar to Example \ref{exam2}, the ANOVA based on the incomplete PB design showed that the significant factors influencing the dependent variable are factors 1-6 and 9 in Example \ref{exam4}.  

By analyzing the four applications above, complete and incomplete PB designs can quickly screen for key factors using traditional statistical methods such as regression or ANOVA, which often focus on the association, correlation, and prediction. However, these methods may not provide cause-and-effect relationships from randomization-based experimental data.  Causal inference is a methodology for drawing causal conclusions about changes in outcomes from observational data (\cite{rubin1974estimating} and \cite{splawa1990application}). Up to now, outstanding advances in causal inference have been achieved across various applications, such as medicine (\cite{aida2026oral}), economics (\cite{wang2025registration}), and transportation (\cite{graham2025causal}), and so forth.  \cite{dasgupta2015causal}  first proposed a causal inference framework based on potential outcomes for $2^{K}$ full factorial designs. Building on this work, numerous studies on $2^{K}$ full factorial designs can be referred to \cite{lu2016randomization,lu2019improved,lu2019sharpening}, \cite{li2020rerandomization},
\cite{zhao2022regression}, \cite{blackwell2023noncompliance}, and \cite{shi2025forward}. Due to resource constraints, such as cost and time, it is difficult to implement $2^{K}$ full factorial designs.  \cite{pashley2023causal} studied causal inference of $2^{K-p}$ fractional factorial designs. To our knowledge, causal inference has been extended to matched-pairs factorial designs (\cite{lu2017randomization}),  blocked designs (\cite{hartman2024improving}, \cite{liu2024randomization}, \cite{zhu2025design}), and split-plot designs (\cite{mukerjee2022causal}, \cite{zhao2022reconciling}).

Although PB designs are two-level factorial designs, most of them are neither $2^K$ full factorial designs nor $2^{K-p}$ fractional factorial designs, such as Examples \ref{exam2}-\ref{exam4}. However, the related research on causal inference is limited to PB designs, and existing results cannot be directly applied to them. Based on this, this paper aims to analyze the causal inference of PB designs and solve the issue in the above four examples.} The rest of the paper is organized as follows. In Section \ref{sec2}, we review the definitions of PB designs and introduce the causal interpretation. Section \ref{sec3} derives the Neymanian estimator of causal effects for PB designs. Section \ref{sec4} performs a sharp null hypothesis test and constructs a Fisherian confidence interval using an algorithm. Section \ref{sec5} reapplies these four examples from a causal inference perspective. A brief conclusion is presented in Section \ref{sec6}.

\section{Causal effects of PB designs}\label{sec2}
 In this section, we first define the causal effects of various order factors in PB designs via potential outcomes, and then introduce a random-assignment mechanism to allocate experimental units.

 PB designs are a type of two-level orthogonal design with $N-1$ factors in $N$ runs, denoted by $OA(N, 2^{N-1})$. 
Let the $N-1$ factor be indexed by $k(k=1,\ldots, N-1)$, and ${\bf z}_j$ be the $j$th treatment combination with $(N-1)$-dimensional vector. Denote ${\bf z}_j=(z_{j1},\ldots, z_{j(N-1)})$, where $z_{jk}$ is the level of $k$th factor in the $j$th treatment combination for $j=1,\ldots,N$ and $k=1,\ldots,N-1$. Without loss of generality, each $z_{jk}$ takes one of two levels: 1 or -1. Here, $1$ indicates a high level and $-1$ indicates a low level. Following \cite{pashley2023causal}, let ${\bf g}_k(k=1,2,\ldots, N-1)$ be a  column vector with $N$-dimension, corresponding to the $k$th factor of the PB design. The $l$-order interaction effects are computed by the dot product of $l$ effects, i.e. ${\bf g}_{k_1k_2\cdots k_l}={\bf g}_{k_1} {\bf g}_{k_2} \cdots {\bf g}_{k_l}$ for $k_1,\ldots, k_l\in \{1,\ldots, N-1\}$. Denote   
\({\bf g}_{0}=(1,1,\ldots, 1)\). \cite{deng1999minimum} applied the $J$-characteristics method for identifying nonregular designs. Thus, for any order $k=1,2,12,\ldots, 12\cdots(N-1)$, a PB design is regular if the values of its 
\({\bf g}_{0}^{T}{\bf g}_{k}\) are either 0 or \(\pm N\), and is nonregular if 
\(0<|{\bf g}_{0}^{T}{\bf g}_{k}|<n\) for at least one $k$. 

Under the stable unit treatment value assumption (SUTVA)  (\cite{rubin1980randomization}), each ${\bf z}_j$ is replicated \(n_j(\geq 2)\) times, resulting in a total of \(n=N\sum_{j=1}^N n_j\) units. Let $Y_i({\bf z}_j)$ be the potential outcome of $i$th unit under ${\bf z}_j$ for $i=1,\ldots, n,j=1,\ldots, N$. Then, the average potential outcome under ${\bf z}_j$ is 
$$\bar{Y}({\bf z}_{j})=\frac{1}{n}\sum_{i=1}^{n} Y_{i}({\bf z}_{j}),j=1,2\ldots, N.$$ 
Denote ${\bf Y}_i=(Y_i({\bf z}_1), Y_i({\bf z}_2),\ldots, Y_i({\bf z}_N))^{\top}(i=1,\ldots,N)$ and $\bar{\bf Y}=(\bar{Y}({\bf z}_1), \bar{Y}({\bf z}_2),\ldots, \bar{Y}({\bf z}_N))^{\top}$. Following \cite{dasgupta2015causal}, we extend the definitions of causal effects to those of PB designs. The main causal effect of each factor is the difference in average outcomes between the two halves of the potential outcomes. Denote $\mathcal{M}_1=\{1,2,\ldots, N-1\}$ as the set of all main effects. Then, the main causal effect of factor $k$ for unit $i$ is expressed by 
\begin{eqnarray}\label{tik}
\tau_i(k)=\frac{2}{N} {\bf g}_{k}^{\top}{\bf Y}_i, k\in \mathcal{M}_1,i=1,\ldots,N.
\end{eqnarray}
Let $\tau_i(0)=(1/N) g_{0}^{\top}{\bf Y}_i$ denote the mean of all ${\bf Y}_i$, and let $\mathcal{M}_2=\{kk': k,k'\in \mathcal{M}_1, k\neq k'\}$ as the set of the interaction effects between two factors. Since there are totally ${N-1\choose 2}$ two-factor interaction effects, the cardinality of the set $\mathcal{M}_2$ is  ${N-1\choose 2}$. Thus, the causal effect of two-factor interaction \(kk'\) for unit $i$ is
$$\tau_i(kk')=\frac{2}{N} {\bf g}_{kk'}^{\top}{\bf Y}_i, kk'\in \mathcal{M}_2, i=1,\ldots, N.$$ 
Denote $\mathcal{M}_l=\{k_1k_2\cdots k_l: k_1, k_2, \ldots, k_l \mbox{ are different and }  k_1, k_2, \ldots, k_l\in \mathcal{M}_1\}$ for $l=2,3\ldots, N-1$. Especially, $\mathcal{M}_l=\mathcal{M}_1$ for $l=1$. In general, the causal effect of an $l$-order interaction $k_1k_2\cdots k_l$ is defined: 
\begin{eqnarray}\label{tikl}
\tau_i(k_1k_2\cdots k_l)=\frac{2}{N} {\bf g}_{k_1k_2\cdots k_l}^{\top}{\bf Y}_i, k_1k_2\cdots k_l\in \mathcal{M}_l,l=1,2,\ldots, N-1, i=1,\ldots, N.\end{eqnarray}
In a finite population, the average causal effect of the $l$-order interaction $k_1k_2\cdots k_l$ is expressed as
\begin{eqnarray}\label{tk}
\tau(k_1k_2\cdots k_l)=\frac{1}{n}\sum_{i=1}^n\tau_i(k_1k_2\cdots k_l)=\frac{2}{N} {\bf g}_{k_1k_2\cdots k_l}^{\top}\left(\frac{1}{n}\sum_{i=1}^n{\bf Y}_i\right)=\frac{2}{N}{\bf g}_{k_1k_2\cdots k_l}^{\top}\bar{\bf Y}
\end{eqnarray}
for $ k_1k_2\cdots k_l\in \mathcal{M}_l(l=1,2,\ldots, N-1)$, where $\bar{\bf Y}=\frac{1}{n}\sum_{i=1}^n {\bf Y}_i$. 
for $ k_1k_2\cdots k_l\in \mathcal{M}_l(l=1,2,\ldots, N-1)$, where $\bar{\bf Y}=\frac{1}{n}\sum_{i=1}^n {\bf Y}_i$.  Note that 
\(\tau(0)=(1/N) g_{0}^{\top} \bar{\bf Y}\). For convenience, we present the above definitions in matrix form. Denote ${\bf G}=(g_0,g_1,\ldots,g_{N-1})$, and \(\tau=(2\tau(0),\tau(1),\tau(2),\ldots,\tau(N-1))^{\top}\). 
Therefore, the relationship between the factorial effects and the average potential outcomes is $\tau=\frac{2}{N}G^{\top}\bar{\bf Y}.$
Because of orthogonality, the inverse of $G$ is \(G^{-1}=\frac{1}{N} G^{T}\), it follows that $\bar{\bf Y}=\frac{1}{2} G\tau$.
\begin{rem}\label{rem}
In addition to geometric and non-geometric PB designs, the definitions (\ref{tikl}) and (\ref{tk}) can be applied to two-level orthogonal arrays. Since PB designs aim to screen out a few truly key factors from numerous influencing factors, we focus hereafter on the main causal effect and the average causal effect of the factors $1,\ldots, N-1$.      
\end{rem}

\begin{table}[h]
\setlength{\tabcolsep}{2pt}
\caption{12-run Plackett-Burman design.}\label{12-run Plackett-Burman design.}%
\label{tab1}
\begin{tabular*}{\linewidth}{@{\extracolsep{\fill}}crrrrrrrrrrrcccccc@{}}
\hline
 \multirow{2}{*}{Treatment} &\multicolumn{11}{c}{Factor ($k$)}&\multicolumn{4}{c}{Potential outcome ($Y_i({\bf z}_j)$)}&\multirow{2}{*}{Average($\bar{Y}({\bf z}_j)$)} \\ \cmidrule(l){2-12}\cmidrule(l){13-16}
combination (${\bf z}_j$)& 1 & 2 & 3 & 4 & 5 & 6 & 7 & 8 & 9 & 10 & 11 &1&2&$\cdots$&$n$&\\
    \hline
    ${\bf z}_1$   & 1 & 1 & -1 & 1 & 1 & 1 & -1 & -1 & -1 & 1 & -1&$Y_1({\bf z}_1)$&$Y_2({\bf z}_1)$&$\cdots$&$Y_n({\bf z}_1)$&$\bar{Y}({\bf z}_1)$ \\
    ${\bf z}_2$  & -1 & 1 & 1 & -1 & 1 & 1 & 1 & -1 & -1 & -1 & 1 &$Y_1({\bf z}_2)$&$Y_2({\bf z}_2)$&$\cdots$&$Y_n({\bf z}_2)$&$\bar{Y}({\bf z}_2)$ \\
    ${\bf z}_3$   & 1 & -1 & 1 & 1 & -1 & 1 & 1 & 1 & -1 & -1 & -1&$Y_1({\bf z}_3)$&$Y_2({\bf z}_3)$&$\cdots$&$Y_n({\bf z}_3)$ &$\bar{Y}({\bf z}_3)$ \\
    ${\bf z}_4$   & -1 & 1 & -1 & 1 & 1 & -1 & 1 & 1 & 1 & -1 & -1&$Y_1({\bf z}_4)$&$Y_2({\bf z}_4)$&$\cdots$&$Y_n({\bf z}_4)$&$\bar{Y}({\bf z}_4)$  \\
    ${\bf z}_5$  & -1 & -1 & 1 & -1 & 1 & 1 & -1 & 1 & 1 & 1 & -1&$Y_1({\bf z}_5)$&$Y_2({\bf z}_5)$&$\cdots$&$Y_n({\bf z}_5)$ &$\bar{Y}({\bf z}_5)$ \\
    ${\bf z}_6$  & -1 & -1 & -1 & 1 & -1 & 1 & 1 & -1 & 1 & 1 & 1&$Y_1({\bf z}_6)$&$Y_2({\bf z}_6)$&$\cdots$&$Y_n({\bf z}_6)$&$\bar{Y}({\bf z}_6)$  \\
    ${\bf z}_7$  & 1 & -1 & -1 & -1 & 1 & -1 & 1 & 1 & -1 & 1 & 1 &$Y_1({\bf z}_7)$&$Y_2({\bf z}_7)$&$\cdots$&$Y_n({\bf z}_7)$&$\bar{Y}({\bf z}_7)$ \\
    ${\bf z}_8$  & 1 & 1 & -1 & -1 & -1 & 1 & -1 & 1 & 1 & -1 & 1&$Y_1({\bf z}_8)$&$Y_2({\bf z}_8)$&$\cdots$&$Y_n({\bf z}_8)$&$\bar{Y}({\bf z}_8)$  \\
    ${\bf z}_9$  & 1 & 1 & 1 & -1 & -1 & -1 & 1 & -1 & 1 & 1 & -1&$Y_1({\bf z}_9)$&$Y_2({\bf z}_9)$&$\cdots$&$Y_n({\bf z}_9)$ &$\bar{Y}({\bf z}_9)$ \\
    ${\bf z}_{10}$ & -1 & 1 & 1 & 1 & -1 & -1 & -1 & 1 & -1 & 1 & 1&$Y_1({\bf z}_{10})$&$Y_2({\bf z}_{10})$&$\cdots$&$Y_n({\bf z}_{10})$&$\bar{Y}({\bf z}_{10})$  \\
    ${\bf z}_{11}$  & 1 & -1 & 1 & 1 & 1 & -1 & -1 & -1 & 1 & -1 & 1 &$Y_1({\bf z}_{11})$&$Y_2({\bf z}_{11})$&$\cdots$&$Y_n({\bf z}_{11})$&$\bar{Y}({\bf z}_{11})$ \\
    ${\bf z}_{12}$ & -1 & -1 & -1 & -1 & -1 & -1 & -1 & -1 & -1 & -1 & -1&$Y_1({\bf z}_{12})$&$Y_2({\bf z}_{12})$&$\cdots$&$Y_n({\bf z}_{12})$&$\bar{Y}({\bf z}_{12})$ \\ 
    \hline
    $ $ &${\bf g}_1$&${\bf g}_2$&${\bf g}_3$&${\bf g}_4$&${\bf g}_5$&${\bf g}_6$&${\bf g}_7$&${\bf g}_8$&${\bf g}_9$&${\bf g}_{10}$&${\bf g}_{11}$&${\bf Y}_1$&${\bf Y}_2$&$\cdots$&${\bf Y}_n$&$\bar{\bf Y}$\\\hline
\end{tabular*}
\end{table}

\begin{exam}{\rm For a PB design with 12 runs in Table \ref{tab1}, there are 11 factors denoted by 1,\ldots, 11, corresponding to the 11-dimension vectors ${\bf g}_k(k=1,\ldots, 11)$. From the definition (\ref{tik}),  the causal effect of factor $1$ for unit $i(i=1,\ldots, n)$ is 
\begin{align*}\tau_i(1)=\frac{2}{12}{\bf g}_{1}^{\top}{\bf Y}_i=&\frac{1}{6}(Y_i({\bf z}_1)-Y_i({\bf z}_2)+Y_i({\bf z}_3)-Y_i({\bf z}_4)-Y_i({\bf z}_5)-Y_i({\bf z}_6)+Y_i({\bf z}_7)+Y_i({\bf z}_8)+Y_i({\bf z}_9)\\
&-Y_i({\bf z}_{10})+Y_i({\bf z}_{11})-Y_i({\bf z}_{12})), i=1,\ldots, n.
\end{align*}
Similarly, the causal effect of other factor $k$ is $\tau_i(k)=\frac{1}{6}{\bf g}_k^{\top}{\bf Y}_i(k=2,\ldots,11)$. Thus, by the definition (\ref{tk}), the average causal effect of factor 1 is written as
\begin{align*}
\tau(1)=\frac{1}{n}\sum_{i=1}^{n}\tau_i(1)= &\frac{1}{6}(\bar{Y}({\bf z}_1)-\bar{Y}({\bf z}_2)+\bar{Y}({\bf z}_3)-\bar{Y}({\bf z}_4)-\bar{Y}({\bf z}_5)-\bar{Y}({\bf z}_6)+\bar{Y}({\bf z}_7)+\bar{Y}({\bf z}_8)+\bar{Y}({\bf z}_9)\\
&-\bar{Y}({\bf z}_{10})+\bar{Y}({\bf z}_{11})-\bar{Y}({\bf z}_{12})),
\end{align*}  
and those of other factors is $\tau(k)=\frac{1}{6}{\bf g}_k^{\top}\bar{\bf Y}(k=2,\ldots,11)$, where $\bar{\bf Y}=(\bar{Y}({\bf z}_1), \bar{Y}({\bf z}_2),\ldots, \bar{Y}({\bf z}_{12}))^{\top}$.
}  
\end{exam}

To estimate the causal effects (\ref{tik}), we need to study the relationship between treatment combination, experimental units, and observed data. For any PB design with $N$ runs, we define random assignment variables  of the $i$th $(i=1,\ldots,n)$ units for the $j$th $(j=1,\ldots, N)$ treatment combination ${\bf z}_j$ as follows
\[
W_i({\bf z}_j) = 
\begin{cases}
    1 ,& \text{if the $i$th unit is assigned to \({\bf z}_j\)},  \\
    0, & \text{otherwise}
\end{cases}
\]
with $P(W_i({\bf z}_j)=1)=\binom{n-1}{n_j-1}/\binom{n}{n_j}=n_j/n$. Under complete randomization, each \({\bf z}_j\) is allocated 
\(n_j\) units. Hence, the observed outcome for the $i$th unit is defined as \(Y_{i}^{obs}=\sum_{i=1}^{n} W_{i}({\bf z}_{j}) Y_{i}({\bf z}_{j})\). Thus, the average observed outcome is
\begin{eqnarray}\label{Yobs}
\bar{Y}^{obs}({\bf z}_{j})=\frac{1}{n_j} \sum_{i=1}^{n} W_{i}({\bf z}_{j}) Y_{i}({\bf z}_{j})=\frac{1}{n_j} \sum_{i: W_{i}(z_{j})=1}^{n} Y_{i}^{obs}, j=1,2\ldots,N.
\end{eqnarray}
Denote \(\bar{\bf Y}^{obs}=(\bar{Y}^{obs}({\bf z}_{1}), \bar{Y}^{obs}({\bf z}_{2}), \ldots, \bar{Y}^{obs}({\bf z}_{N}))\). Based on the observed data $\bar{\bf Y}^{obs}$, we conduct the statistical inference for causal effects in the following section.

\section{Causal inference}\label{secsthree}
In this section, we provide the estimators and confidence intervals for causal effects and focus on the difference between Neymanian and Fisherian variances.
\subsection{Neymanian estimators}\label{sec3}
The Neymanian estimator is a foundational approach in causal inference for randomized experiments. \cite{splawa1990application} pointed out that $\bar{\bf Y}^{obs}$ is an unbiased estimator of $\bar{\bf Y}$ under complete randomization. By replacing $\bar{\bf Y}^{obs}$ with $\bar{\bf Y}$, the estimator of $\tau(k)$ of PB designs with $N$ runs is
\begin{eqnarray}\label{3.1}
\hat{\tau}(k)=\frac{2}{N} {\bf g}_{k}^{\top}\bar{\bf Y}^{obs}, k\in \mathcal{M}_1.
\end{eqnarray}

\cite{dasgupta2015causal} derived an unbiased estimator and variance of casual effects for $2^{K}$ full factorial designs under balanced conditions. \cite{pashley2023causal} derived the properties of casual estimators for $2^{K-p}$ fractional factorial designs under unbalanced conditions. 
Following their approaches, we directly derive the expectation, variance, and covariance of the estimator $\hat{\tau}(k)$ for PB designs. 

\begin{theorem}\label{thm1}{\rm  
For $k\in \mathcal{M}_1$, we have
$E(\hat{\tau}(k))={\tau}(k)$, and
\begin{align}\label{var1}
\mbox{Var}_N\left(\hat{\tau}(k)\right)=\frac{4}{N^2} \sum_{j=1}^{N} \frac{1}{n_j}S^{2}\left({\bf z}_{j}\right) -\frac{1}{n} S_{k}^{2},
\end{align}
where $S^{2}\left({\bf z}_j\right)=\frac{1}{n-1}\sum_{i=1}^{n}(Y_i({\bf z}_j))-\bar{Y}({\bf z}_j))^2$ and $S_{k}^2=\frac{1}{n-1}\sum_{i=1}^{n}(\tau_i(k)-\tau(k))^2.$}
\end{theorem}

Theorem \ref{thm1} reveals that $\hat{\tau}(k)$ is an unbiased estimator of ${\tau}(k)$, and Var$(\hat{\tau}(k))$ shows the fluctuation of $\hat{\tau}(k)$. Denote
$$S_{k,k'}^2=\frac{1}{n-1}\sum_{i=1}^{n}(\tau_i(k)-\tau(k))(\tau_i(k')-\tau(k')), k,k'\in \mathcal{M}_1.$$
It is necessary to study the covariance \mbox{Cov}$(\hat{\tau}(k),\hat{\tau}(k'))$.

\begin{theorem}\label{thm2}
{\rm 
For any $k,k'\in \mathcal{M}_1$,  we have
\[\mbox{Cov}^{}( \hat {\tau }(k),\hat {\tau }(k^{\prime })) =\frac{4}{N^2}\left[ \sum _{j:{\bf g}_{kj}={\bf g}_{k^{\prime }j}}\frac{1}{n_j}S^{2}({\bf z}_{j})-\sum _{j:{\bf g}_{kj}\neq {\bf g}_{k^{\prime }j}}  \frac{1}{n_j}S^{2}\left ({\bf z}_{j}\right) \right] -\frac {1}{n}S^{2}_{k,k^{\prime}}.\]
}
\end{theorem}

Ignoring the population variance of \(\tau(k)\) for all units,  we directly derive the estimates for PB designs under unbalanced conditions by Theorem \ref{thm2}.

\begin{corollary}\label{VarNey}{\rm 
 The conservative Neymanian variance estimator of \(\tau(k)\) satisfies
\begin{align}\label{cor}
\widehat{\text{Var}}_{\text{N}}(\hat{\tau}(k)) = \frac{4}{N^2}\sum_{j=1}^{N}\frac{1}{n_j} s^{2}({\bf z}_{j}), \qquad E[\widehat{\text{var}}_{\text{N}}\left(\hat{\tau}(k)\right)]  = \frac{4}{N^2}\sum_{j=1}^{N}\frac{1}{n_j} S^{2}\left({\bf z}_{j}\right),
\end{align}
 where  $s^{2}({\bf z}_j)=\frac{1}{n_j-1}\sum_{i=1}^{n}(Y^{obs}_i({\bf z}_j))-\bar{Y}^{obs}({\bf z}_j))^2$ is an unbiased estimator of \(S^{2}({\bf z}_{j})\).
}
\end{corollary}

Compared with (\ref{var1}) and (\ref{cor}), it follows that the overestimated portion is 
$
E[\widehat{\text{Var}}_{\text{N}}\left(\hat{\tau}(k)\right)]-\widehat{\text{Var}}(\tau(k))=S_{k}^2/n.
$
\begin{theorem}\label{thm3}{\rm If the three conditions hold: 
  (i) $S^{2}(z_{j})$ and $S_{k, k'}^{2}$ have limiting values, (ii) $\mbox{max} _{1\leq j\leq N} \mbox{max} _{1\leq i \leq n}(Y_{i}(z_{j})-\bar{Y}(z_{j}))^{2}/n \to 0$, (iii) $n_j/{n}$ has positive limit values, then
$\sqrt{n}\bigl(\hat{\tau}-E(\hat{\tau})\bigr)\xrightarrow{d}N(0,V),n \to \infty,$
where $V$ is the limit values of $n\mbox{Var}(\hat{\tau})$.}
\end{theorem}
\begin{proof} 
It follows from Theorem 5 in \cite{li2017general}.
\end{proof}

 Based on Theorem \ref{thm3}, the marginal normal distribution  of the $k$th component as
$\sqrt{n}\bigl(\hat{\tau}(k)-E(\hat{\tau}k))\bigr)\xrightarrow{d}N(0,V_{kk})$ for any $k\in\mathcal{M}_1$
where $V_{kk}$ is the $k$th diagonal element of $V$, which is the limiting value of $n\mbox{Var}(\hat{\tau}{(k)})$. Since $\mbox{Var}(\hat{\tau}{(k)})$ contains the intractable quantity $S_k^2$, we general use $\text{Var}_{\text{Ney}}(\hat{\tau}(k))$ from Corollary \ref{VarNey}. Further, an asymptotically conservative $100(1-\alpha)\%$ confidence interval for $\tau(k)$ is given by
\begin{align}\label{NeyCI}
\hat \tau(k)+z_{1-\alpha/2}\sqrt{{\text{var}}_{\text{Ney}}(\hat{\tau}(k))},
\end{align}
where $z_{1-\alpha/2}$ is the $1-\alpha/2$ quantile of the standard normal distribution.

Through Theorems \ref{thm1}-\ref{thm3}, we note that Neymanian inference focuses on the statistical inference of the average causal effects at the population level. In practice, it is also necessary to study that at the individual level. To solve the problem, we analyze the Fisherian inference based on the potential outcomes in the following subsection.
\subsection{Fisherian test and confidence interval}\label{sec4}
This subsection primarily investigates Fisherian test and confidence intervals for causal effects $\tau_i(k)(k=1,\ldots, N_1,i=1,\ldots, N)$ in the PB design. 
We first conduct a sharp null hypothesis that the individual-level effects $\tau_{i}(k)$ are exactly equal to $\eta_k$ in unit $i$ as follows:
\begin{align}\label{H0}
H_0^{\eta}:  \tau_{i}(k) = \eta_k,\quad k=1,\ldots, N-1, i=1,\ldots, N.
\end{align}

By (\ref{tik}) and the orthogonality of \(G\), namely \(G^{-1}=\frac{1}{N}G^\top\), the relationship between the vectors \({\bf Y}_i=(Y_i({\bf z}_1),
Y_i({\bf z}_2),\ldots,Y_i({\bf z}_{N-1}))^\top\) and \(\bm{\tau}_i=(2\tau_i(0),\tau_i(1),\tau_i(2),\ldots,\tau_i(N-1))^\top\) can be derived ${\bm \tau}_i=\frac{2}{N}G^T{\bf Y}_i$. Thus, we have ${\bf Y}_i=\frac{1}{2}G\bm{\tau}_i$. Let ${\bf z}_j=(z_{j1},z_{j2},\cdots, z_{j(N-1)})(j=1,\ldots, N)$ be the $j$th treatment combination. Thus,
\begin{align}
\begin{pmatrix}
Y_i(\mathbf{z}_1)\\
Y_i(\mathbf{z}_2)\\
\vdots\\
\vdots\\
Y_i(\mathbf{z}_{N-1})
\end{pmatrix}
=
\frac{1}{2}\begin{pmatrix}
1 & z_{11} & z_{12} & \cdots & z_{1(N-1)}\\
1 & z_{21} & z_{22} & \cdots & z_{2(N-1)}\\
\vdots & \vdots & \vdots & \ddots & \vdots\\
\vdots & \vdots & \vdots & \ddots & \vdots\\
1 & z_{N1} & z_{N2} & \cdots & z_{N(N-1)}
\end{pmatrix}
\begin{pmatrix}
2\tau_i(0)\\[1mm]
\tau_i(1)\\
\tau_i(2)\\
\vdots\\
\tau_i(N-1)
\end{pmatrix},
\end{align}
Denote $\mu_i=\tau_i(0)$ as a baseline value of unit $i$. Let ${\bf z}^{obs}_i=(z^{obs}_{i1},z^{obs}_{i2},\cdots, z^{obs}_{i(N-1)})$ be the observed treatment combination randomly assigned to unit $i$. For each $i$th unit, only the observed outcome $Y^{obs}_i$ is known. 
Denote a known vector ${\bm \eta}=(\eta_1,\ldots, \eta_{N-1})$. 
Under $H_0^\eta: \tau_{i}(k) = \eta_k(k=1,\ldots, N-1)$, we have
\begin{align}\label{imputeion1}
\mu_i=\tau_i(0)=Y_i^{obs}-\frac{1}{2}{\bf z}^{obs}_{i}{\bm \eta}=Y_i^{obs}-\frac{1}{2}\sum_{k=1}^{N-1}z^{obs}_{ik}\eta_k
\end{align}
to impute the remaining unobserved potential outcomes expressed by
\begin{align}\label{imputeion2}
Y_i^{mis}({\bf z}_j)=\mu_i+\frac{1}{2}{\bf z}_{j}{\bm \eta}=\mu_i+\frac{1}{2}\sum_{k=1}^{N-1}z_{jk}\eta_k, i=1,\ldots, n, j=1,\ldots, N.
\end{align}

After imputing all missing potential outcomes $\{Y_i^{mis}({\bf z}_j)\}$, we need to construct a statistic and test $H_0^{\bm\eta}$ of the $k$th factor in given $\bar{\bf Y}^{obs}$ and random assignment variable $W$, denoted by $\hat\tau(k, W, \bar Y^{obs})$. In general, $\hat\tau(k, W, \bar Y^{obs})=\hat{\tau}(k)=\frac{2}{N} {\bf g}_{k}^{\top}\bar{\bf Y}^{obs}$. From the assignment variable $W$, the total number of possible schemes is $\frac{n!}{\prod_{j=1}^{N} n_j!}$, which is too large to enumerate completely.  Thus, we choose the $B$ random permutations of $W$, denoted by $W^b$.  Define a \(p\)-value function  $p(\eta_k)$ of \(\eta_k\). We then use the Monte Carlo method, proposed by \cite{ding2024first}, to estimate the value $p(\eta_k)$ of the $k$th factor as follows:
\begin{align}\label{pvalu}
    \hat p(\eta_k)
=\frac{1}{B}\sum_{b=1}^{B}\mathbf 1\!\left\{\hat\tau(k,W^b,\bar Y^{b,obs})\geq \hat\tau(k,W,\bar Y^{obs})\right\}.
\end{align}
Given $0<\alpha<1$, if $\hat p(\eta_k) \geq\alpha$, we accept $H^\eta_0$. Otherwise, reject $H^\eta_0$. 

For the Fisherian inference, the new observed outcomes obtained from the imputed complete table $Y_i({\bf z}_j)$ under $W^b$ are defined as
$Y^{b, obs}_i=\sum_{i=1}^{n} W^b_{i}({\bf z}_{j}) Y_{i}({\bf z}_{j})$,
and the corresponding average observed outcome is
\begin{align}\label{observed outcome1}
    \bar{Y}^{b,obs}({\bf z}_{j})=\frac{1}{n_j} \sum_{i=1}^{n} W^b_{i}({\bf z}_{j}) Y_{i}({\bf z}_{j}), j=1,2\ldots, N.
\end{align}

\cite{luo2021leveraging} showed that, for the binary treatment, a statistic constructed from the difference in means induces monotone randomization $p$-values. Based on this, we prove the monotonicity of $\hat p(\eta_k)$ below.
\begin{prop}\label{prop}
{\rm
In (\ref{pvalu}), $
 \hat p(\eta_k)
=\frac{1}{B}\sum_{b=1}^{B}\mathbf 1\!\left\{  \hat\tau(k,W,\bar Y^{obs})\leq a+b\eta_k)\right\},
$ and $\hat{p}(\eta_k)$ is monotonically with respect to $\eta_k$.}
\end{prop}
\begin{proof}
Denote \(\bar{\bf Y}^{b,obs}=(\bar{Y}^{b,obs}({\bf z}_{1}), \bar{Y}^{b,obs}({\bf z}_{2}), \ldots, \bar{Y}^{b,obs}({\bf z}_{N}))\).  First, substituting (\ref{imputeion1}) into (\ref{imputeion2}) yields:
\begin{align}
Y_i({\bf z}_j)=Y_i^{obs}+\frac{1}{2}\sum_{k=1}^{N-1}(z_{jk}-z^{obs}_{ik})\eta_k.\notag
\end{align}
The average observed outcome can be simplified to the following equation
\begin{align}
{\bar Y}^{b,obs}({\bf z}_j)&=\frac{1}{n_j}\sum_{i=1}^{n}W^b_i({\bf z}_j)\left(Y_i^{obs}+\frac{1}{2}\sum_{k=1}^{N-1}(z_{jk}-z^{obs}_{ik})\eta_k\right)\notag\\&=\frac{1}{n_j}\sum_{i=1}^{n}W^b_i({\bf z}_j)\left(Y_i^{obs}+\frac{1}{2}\sum_{k'\neq k}^{N-1}(z_{jk'}-z^{obs}_{ik'})\eta_{k'}\right)+\frac{1}{2}(z_{jk}-\frac{1}{n_j}\sum_{i=1}^{n}W^b_i({\bf z}_j)z^{obs}_{ik})\eta_k.\notag
\end{align}

Denote ${\bar z}^{obs}_{k}({\bf z}_j,W^b)=\frac{1}{n_j}\sum_{i=1}^{n}W^b_i({\bf z}_j)z^{obs}_{ik}$, which is the average value of the $k$th component of the ${\bf z}_j$ across all units assigned to ${\bf z}_j$ under $W^b$. Thus,
\begin{align}
{\bar Y}^{b,obs}({\bf z}_j)=\frac{1}{n_j}\sum_{i=1}^{n}W^b_i({\bf z}_j)\left(Y_i^{obs}+\frac{1}{2}\sum_{k'\neq k}^{N-1}(z_{jk'}-z^{obs}_{ik'})\eta_{k'}\right)+\frac{1}{2}(z_{jk}-{\bar z}^{obs}_{k}({\bf z}_j,W^b))\eta_k.\notag
\end{align}
It follows that $\hat\tau(k,W^b,\bar{\bf Y}^{b,obs})$ can be expressed as
\begin{align}
\hat\tau(k,W^b,\bar{\bf Y}^{b,obs})=&\frac{2}{N}g_k\bar{\bf Y}^{b,obs}=\frac{2}{N}\sum_{j=1}^{N}z_{jk}\bar{Y}^{b,obs}({\bf z}_j)\notag\\=&\frac{2}{N}\sum_{j=1}^{N}z_{jk}\left(\frac{1}{n_j}\sum_{i=1}^{n}W^b_i({\bf z}_j)\left(Y_i^{obs}+\frac{1}{2}\sum_{k'\neq k}^{N-1}(z_{jk'}-z^{obs}_{ik'})\eta_{k'}\right)\right)
+\frac{1}{N}\sum_{j=1}^{N}z_{jk}(z_{jk}-{\bar z}^{obs}_{k}({\bf z}_j,W^b))\eta_k\notag\\
=&a+b\eta_k\notag,
\end{align}
where 
\begin{align*}
a=&\frac{2}{N}\sum_{j=1}^{N}z_{jk}\left(\frac{1}{n_j}\sum_{i=1}^{n}W^b_i({\bf z}_j)\left(Y_i^{obs}+\frac{1}{2}\sum_{k'\neq k}^{N-1}(z_{jk'}-z^{obs}_{ik'})\eta_{k'}\right)\right),\\
b=&\frac{1}{N}\sum_{j=1}^{N}z_{jk}(z_{jk}-{\bar z}^{obs}_{k}({\bf z}_j,W^b))=\frac{1}{N}\sum_{j=1}^{N}(1-z_{jk}{\bar z}^{obs}_{k}({\bf z}_j,W^b)).\end{align*}
Hence, 
$ \hat p(\eta_k)
=\frac{1}{B}\sum_{b=1}^{B}\mathbf 1\!\left\{a+b\eta_k\geq  \hat\tau(k,W,\bar Y^{obs})\right\}.
$ Since $0\leq z_{jk}\leq1$ and $0\leq{\bar z}^{obs}_{k}({\bf z}_j,W^b)\leq1$, it follows that \(0\le z_{jk}{\bar z}^{obs}_{k}({\bf z}_j,W^b)\le 1\), and hence $b\ge 0$. Thus, \(\hat\tau(k,W^b,\bar Y^{b,obs})\) is a nondecreasing function of \(\eta_k\). This means that $\hat p(\eta_k)$ is nondecreasing in $\eta_k$. 
\end{proof}

Proposition \ref{prop} proves the monotonicity of $\hat p(\eta_k)$ with respect to \(\eta_k\), based on which we can construct a confidence interval using the bisection method. The Fisherian interval is constructed by repeatedly performing the hypotheses $H_0^\eta: \tau_{i}(k) = \eta_k$. Based on the test statistic $\hat \tau(k, W, \bar Y^{obs})$, we construct 
an interval interval $[\hat \tau(k, W,\bar Y^{obs})-c, \hat \tau(k, W,\bar Y^{obs})+c]$  such that a candidate point $\eta_k$ is generated within this interval to replace the $k$th element of $\bm \eta$, where $c$ is defined as the critical value. Under the $B$ random permutation $W^b$, the $\hat\tau(k, W^b, \bar{Y}^{b, obs})$ is then calculated based on the imputation data (\ref{imputeion1}) and (\ref{imputeion2}), and finally, the $\alpha$ and 1-$\alpha$ \text{ quantile} of the test statistic distribution $\hat p(\eta_k)$ are approximated using the bisection method to determine the confidence interval bounds. Let $L_k$ and $U_k$ denote the lower and upper bounds of the confidence interval, respectively, defined as follows: 
\begin{align*}
    L_k=\sup\{\eta_k:\hat p(\eta_k)\leq\alpha/2\},U_k=\inf\{\eta_k: \hat p(\eta_k)\geq1-\alpha/2\}.
\end{align*}
Then, the confidence interval for the $k$th factor $\tau(k)$ is given by $[L_k, U_k]$. The pseudocode for constructing the confidence interval is provided in Algorithm \ref{alg1} (see Appendix).

\subsection{Comparison of Neymanian and Fisherian variance}
Under the null hypothesis, \cite{ding2017paradox} established the difference between the Neymanian variance and the Fisherian variance in the $2^{K}$ full factorial designs. In this subsection, we analyze the difference between them in the case $H^\eta_{0}: \tau_i(k)=\eta_k$. 
Under $H_0^\eta$, $\mu_i$ is a unit baseline value that is invariant across ${\bf z}_j$. Accordingly, we define $
\bar{\mu}({\bf z}_{j})=\frac{1}{n_j} \sum_{i=1}^{n} W^b_{i}({\bf z}_{j})\mu_i
$ as the average baseline values for units assigned to ${\bf z}_j$ under $W^b$. 
By substituting (\ref{imputeion2}) into (\ref{observed outcome1}), the average potential outcome after imputation is obtained as
\begin{align}\label{Y}
\bar{Y}^{b,obs}({\bf z}_{j})=\frac{1}{n_j} \sum_{i=1}^{n} W^{b}_{i}({\bf z}_{j})\mu_i+\frac{1}{2n_j}\sum_{i=1}^{n} W^{b}_{i}({\bf z}_j){\bf z}_j{\bm \eta}=\bar{\mu}({\bf z}_{j})+\frac{1}{2}{\bf z}_j{\bm \eta}.
\end{align}

\begin{theorem}
Denote $\hat{\tau}^b(k)\equiv\hat\tau(k,W^b,\bar Y^{b,obs})=\frac{2}{N}g_k\bar{\bf Y}^{b,obs}$. The expectation and variance of $\hat{\tau}^b(k)$ in Fisherian inference are given by
\begin{align*}
E_F(\hat{\tau}^b(k)|H^\eta_0)=\eta_k,\quad {\mbox{Var}}_F(\hat{\tau}^b(k)|H^\eta_0)=\frac{4s^{2}}{N^2}\sum_{j=1}^{N}\frac{1}{n_j}.
 \end{align*}
The variance of the baseline values $\mu_i$ is calculated as
$s^2=\frac{1}{n-1}\sum_{i=1}^{n}(\mu_i-\bar{\mu})^2$, where $\bar{\mu}=\sum_{i=1}^{n}{\mu}_i/n$.
\end{theorem}

\begin{proof}
Since $\mu_i$ is invariant under ${\bf z}_j$ that satisfies the conditions of Theorem 5's proof in \cite{ding2017paradox}, it follows that
\begin{align}
&{\mbox{Var}}\left(\bar{Y}^{b,obs}({\bf z}_{j})|H^\eta_0\right)=\mbox{Var}\left(\bar{\mu}({\bf z}_{j})|H^\eta_0\right)= \frac{n-n_j}{n_jn}s^2,\notag\\&
\mbox{Cov}\left(\bar{Y}^{obs,\eta}({\bf z}_{j}),\bar{Y}^{obs,\eta}({\bf z}_{j'})|H^\eta_0\right)=\mbox{Cov}\left(\bar{\mu}({\bf z}_{j}),\bar{\mu}({\bf z}_{j'})|H^\eta_0\right)= -\frac{1}{n}s^2.\label{VaYrobs}
\end{align}
Using the results in (\ref{VaYrobs}), the expressions for the expectation and variance of $\hat{\tau}^b(k)$ under are given as follows:
\begin{align*}
E_F\left(\hat{\tau}^b(k)|H^\eta_0\right)&=\frac{2}{N}{\bf g}_k^TE\left(\bar{Y}^{b,obs}|H^\eta_0\right)=\frac{2}{N}\sum_{j=1}^N{z}_{jk}E\left(\bar{\mu}({\bf z}_j)|H^\eta_0\right)+\frac{1}{N}\sum_{j=1}^N{ z}_{jk}{\bf z}_{j}{\bm \eta}=\frac{1}{N}\sum_{j=1}^N{ z}_{jk}{ z}_{jk}\eta_k=\eta_k,
\\
{\mbox{Var}}_F\left(\hat{\tau}^b(k)|H^\eta_0\right) & 
=\frac{4}{N^2}{\bf g}_{k}^{T} {\mbox{Var}}\left(\bar{\bf Y}^{b,obs}|H^\eta_0\right) {\bf g}_{k}\\& =\frac{4}{N^2}\left[\sum_{j=1}^{N} {\bf g}_{k j}^{2} 
\mbox{Var}\left(\bar{\mu}\left({\bf z}_{j}\right)|H^\eta_0\right)+\sum_{j} \sum_{{j'} \neq j} {\bf g}_{k j} {\bf g}_{k {j'}} {\mbox{Cov}}\left(\bar{\mu}\left({\bf z}_{j}\right), \bar{\mu}\left({\bf z}_{j'}\right)|H^\eta_0\right)\right]
\\&=\frac{4}{N^2}s^{2}\left[\sum_{j=1}^{N}\frac{n-n_j}{nn_j}{\bf g}^2_{k j}-\frac{1}{n} \sum_{j}^{N} \sum_{{j'}\neq j}^{N} {\bf g}_{k j} {\bf g}_{k {j'}}\right]
\\&=\frac{4}{N^2}s^{2}\left[\sum_{j=1}^{N}\frac{1}{n_j}{\bf g}^2_{k j}-\frac{1}{n}\left(\sum_{j=1}^{N}{\bf g}^2_{k j}+\sum_{j}^{N} \sum_{{j'}\neq j}^{N} {\bf g}_{k j} {\bf g}_{k {j'}}\right)\right]\\&=\frac{4}{N^2}s^{2}\left[\sum_{j=1}^{N}\frac{1}{n_j}{\bf g}^2_{k j}-\frac{1}{n}\left(\sum_{j=1}^{N}{\bf g}_{k j}\right)^2\right]=\frac{4}{N^2}\sum_{j=1}^{N}\frac{1}{n_j}s^{2}.
\end{align*}
\end{proof}

\begin{theorem}
Under unbalanced conditions, the difference between ${\mbox{Var}}_{F}(\hat{\tau}^b(k))$ and $\widehat{\mbox{Var}}_{N}(\hat{\tau}(k))$ is given by
\begin{align*} 
{\mbox{Var}}_{F}(\hat{\tau}^b(k))-\widehat{\mbox{Var}}_N(\hat{\tau}(k))=&\frac{4}{N^2(n-1)}\sum_{j=1}^{n}\frac{1}{n_j}\sum_{{j'}=1}^{n}n_{j'}\left(\bar{\mu}({\bf z}_j)-\bar{\mu}\right)^{2}\\
&+\frac{4}{N^2}\left(\sum_{j'=1}^{N}(\frac{(n_{j'}-1)\sum_{j=1}^{N}n^{-1}_j}{(n-1)}-\frac{1}{n_{j'}})s^2({\bf z}_{j’})\right),  
\end{align*}
where $s^{2}({\bf z}_j)=\frac{1}{n_j-1}\sum_{i=1}^{n}(Y_i({\bf z}_j))-\bar{Y}({\bf z}_j))^2=\frac{1}{n_j-1}\sum_{i=1}^{n}(\mu_i-\bar{\mu}({\bf z}_j))^2$.
Under balanced conditions $n_j=r(j=1,\ldots,N)$, each treatment combination is repeated $r$ times, with $r \to \infty$,  
\begin{align} \label{balanced}
&{\mbox{Var}}_{F}(\hat{\tau}^b(k))-\widehat{\mbox{Var}}_N(\hat{\tau}(k))=\frac{2}{N^3r}\sum_{j=1}^{n}\sum_{j’=1}^{n}( \left(\bar{\mu}({\bf z}_j)-\bar{\mu}({\bf z}_{j'})\right)^{2}+o_{p}(r^{-1}).
\end{align}
Thus, the variance of Fisherian is greater than that of Neymanian.
\end{theorem}
\begin{proof}
First,  expanding $s^{2}$ yields that
\begin{align}\label{s^2}
s^{2}=& \frac{1}{n-1}\sum_{i=1}^{n}\left(\mu_i-\bar{\mu}\right)^{2}\notag\\=& \frac{1}{n-1}\sum_{j=1}^{N}\sum_{\left\{i: W_{i}({\bf z}_j)=1\right\}}\left(\mu_i-\bar{\mu}({\bf z}_j)+\bar{\mu}({\bf z}_j)-\bar{\mu}\right)^{2} \notag\\=& \frac{1}{n-1}\sum_{j=1}^{N}\sum_{\left\{i: W_{i}({\bf z}_j)=1\right\}}\left(\mu_i-\bar{\mu}({\bf z}_j)\right)+\frac{1}{n-1}\sum_{j=1}^{n}n_j\left(\bar{\mu}({\bf z}_j)-\bar{\mu}\right)^{2}\notag\\=& \frac{1}{n-1}\sum_{j=1}^{N}(n_j-1)s^{2}({\bf z}_j)+\frac{1}{n-1}\sum_{j=1}^{n}n_j\left(\bar{\mu}({\bf z}_j)-\bar{\mu}\right)^{2}.
\end{align}

To compute ${\mbox{Var}}_{F}(\hat{\tau}^b(k))-\widehat{\mbox{Var}}_N(\hat{\tau}(k))$, we substitute $s^{2}$ in the expression with its form given in Equation (\ref{s^2}), yielding the final expression for ${\mbox{Var}}_{F}(\hat{\tau}^b(k))-\widehat{\mbox{Var}}_N(\hat{\tau}(k))$ as follows:
\begin{align*} 
{\mbox{Var}}_{F}(\hat{\tau}^b(k))-\widehat{\mbox{Var}}_N(\hat{\tau}(k))=&\frac{4}{N^2}\sum_{j=1}^{N}\frac{1}{n_j}s^{2}-\frac{4}{N^2}\sum_{j=1}^{n}\frac{1}{n_j}s^2({\bf z}_j) \\  =&\frac{4}{N^2(n-1)}\left(\sum_{j=1}^{N}\frac{1}{n_j}\left(\sum_{j'=1}^{N}(n_{j'}-1)s^{2}({\bf z}_j)+\sum_{{j'}=1}^{n}n_{j'}\left(\bar{\mu}({\bf z}_j)-\bar{\mu}\right)^{2}\right)-\sum_{j=1}^{n}\frac{1}{n_j}s^2({\bf z}_j)\right)\\  =&\frac{4}{N^2(n-1)}\sum_{j=1}^{n}\frac{1}{n_j}\sum_{{j'}=1}^{n}n_{j'}\left(\bar{\mu}({\bf z}_j)-\bar{\mu}\right)^{2}\\
&+\frac{4}{N^2}\left(\sum_{j'=1}^{N}(\frac{(n_{j'}-1)\sum_{j=1}^{N}n^{-1}_j}{(n-1)}-\frac{1}{n_{j'}})s^2({\bf z}_{j’})\right). 
\end{align*}

In the special case of a balanced design, ${\mbox{Var}}_{F}(\hat{\tau}^b(k))=\frac{4}{Nr}s^{2}$, and $\widehat{\mbox{Var}}_N(\hat{\tau}(k))=\frac{4}{N^2r}\sum_{j=1}^{N}s^2({\bf z}_j)$. As $p \to \infty$, the differences among $n$, $n-1$, $r-1$, and $r$ become negligible. Thus, under a balanced design $s^{2}$ can be simplified as
\begin{align} 
s^{2}=& \frac{1}{n-1}\sum_{i=1}^{n}\left(\mu_i-\bar{\mu}\right)^{2}\notag\\=& \frac{1}{n-1}\sum_{j=1}^{N}\sum_{\left\{i: W_{i}({\bf z}_j)=1\right\}}\left(\mu_i-\bar{\mu}({\bf z}_j)+\bar{\mu}({\bf z}_j)-\bar{\mu}\right)^{2}\notag \\=& \frac{1}{n-1}\sum_{j=1}^{N}\sum_{\left\{i: W_{i}({\bf z}_j)=1\right\}}\left(\mu_i-\bar{\mu}({\bf z}_j)\right)+\frac{r}{n-1}\sum_{j=1}^{n}\left(\bar{\mu}({\bf z}_j)-\bar{\mu}\right)^{2}\notag\\=& \frac{r-1}{n-1}\sum_{j=1}^{N}s^{2}({\bf z}_j)+\frac{r}{n-1}\sum_{j=1}^{n}\left(\bar{\mu}({\bf z}_j)-\bar{\mu}\right)^{2}\notag\\=& \frac{1}{N}\sum_{j=1}^{N}s^{2}({\bf z}_j)+ \frac{1}{N}\sum_{j=1}^{n}\left(\bar{\mu}({\bf z}_j)-\bar{\mu}\right)^{2} + o_p(r^{-1}),
\end{align}
Therefore, under a balanced design, the difference between 
${\mbox{Var}}_{F}(\hat{\tau}^b(k))$ and $\widehat{\mbox{Var}}_{N}(\hat{\tau}(k))$ is given by
\begin{align*} 
{\mbox{Var}}_{F}(\hat{\tau}^b(k))-\widehat{\mbox{Var}}_N(\hat{\tau}(k))&=\frac{4}{Nr}s^{2}-\frac{4}{N^2r}\sum_{j=1}^{N}s^2({\bf z}_j)=\frac{4}{N^2r}(Ns^{2}-\sum_{j=1}^{N}s^2({\bf z}_j))\\ & =\frac{4}{N^2r}(\sum_{j=1}^{n}\left(\bar{\mu}({\bf z}_j)-\bar{\mu}\right)^{2} + o_p(r^{-1}))=\frac{2}{N^3r}\sum_{j=1}^{n}\sum_{j’=1}^{n}( \left(\bar{\mu}({\bf z}_j)-\bar{\mu}({\bf z}_{j'})\right)^{2}+o_{p}(r^{-1}).
\end{align*}
\end{proof}

\section{Empirical analysis of Causal inferences}\label{sec5}

In this section, we primarily use Examples \ref{exam1}-\ref{exam4} to illustrate the causal inference in PB designs. The PB designs can be classified into four cases according to these examples:

Case 1. PB designs are complete and geometric, and the type of observation is mean ± standard deviation, see Example \ref{exam1}.

Case 2. PB designs are complete and non-geometric, and the type of observation is also mean ± standard deviation, see Example \ref{exam2}.

Case 3. PB designs are incomplete and non-geometric, and the type of observation is also mean ± standard deviation, see Example \ref{exam3}.

Case 4. PB designs are incomplete and non-geometric, and each replicate observed outcome is available, see Example \ref{exam4}.

For Cases 1-3, Neymanian inference can be performed for the design type (complete or incomplete, geometric or non-geometric) and any data reporting format. However, Fisherian inference requires that the outcome of each replicate observed outcomes is known; moreover, if the design is incomplete, it must be expanded to a complete design, with the effects of the added factor column assumed to be zero, for example, Case 4. Next, we illustrate the detailed process of the theoretical results using the above examples.

{\bf Example \ref{exam1}.} (Cont.)   \cite{filgueiras2021plackett} provided mean $\pm$ standard deviation of outcomes in Table \ref{exam_table1}.  The PB design with 8 runs, which is complete and geometric.

\begin{table}[!http]
\caption{Original design matrix and  data for Example \ref{exam1}.}
\setlength{\tabcolsep}{1pt}
\label{exam_table1}
\begin{tabular*}{\linewidth}{@{\extracolsep{\fill}}c*{7}{r}cccccc@{}}
\toprule
\multirow{2}{*}{${\bf z}_j$} 
& \multicolumn{7}{c}{Factor} 
& \multicolumn{3}{c}{Fibres} 
& \multicolumn{3}{c}{Fragments}  \\
\cmidrule(lr){2-8}
\cmidrule(lr){9-11}
\cmidrule(lr){12-14}
& 1 & 2 & 3 & 4 & 5 & 6 & 7 &mean $\pm$ SD&$\bar{Y}^{obs}({\bf z}_j)$ &$s^2({\bf z}_j)$ &mean $\pm$ SD&$\bar{Y}^{obs}({\bf z}_j)$ & $s^2({\bf z}_j)$\\
\midrule
${\bf z}_1$ &  1 &  1 &  1 & -1 &  1 & -1 & -1 & 26 $\pm$ 24.1 & 26 &580.81 & 106 $\pm$ 40.4& 106   &  1632.16 \\
${\bf z}_2$ &  1 &  1 & -1 &  1 & -1 & -1 &  1 & 232 $\pm$ 108&232 & 11664.00 &448 $\pm$ 103 &448   & 10609.00 \\
${\bf z}_3$ &  1 & -1 &  1 & -1 & -1 &  1 &  1 & 42 $\pm$ 26.8 &42 &   718.24 &292 $\pm$ 47.6 &292   &  2265.76 \\
${\bf z}_4$ & -1 &  1 & -1 & -1 &  1 &  1 &  1 &253 $\pm$ 269 &253 & 72361.00 &540 $\pm$ 142&540   & 20164.00 \\
${\bf z}_5$ &  1 & -1 & -1 &  1 &  1 &  1 & -1 &216 $\pm$ 84.1 &216 &  7072.81 & 336 $\pm$ 81.7&336   &  6674.89 \\
${\bf z}_6$ & -1 & -1 &  1 &  1 &  1 & -1 &  1 & 90 $\pm$ 85.5 &90 &  7310.25 & 86.7 $\pm$ 27.4 &86.7 &   750.76 \\
${\bf z}_7$ & -1 &  1 &  1 &  1 & -1 &  1 & -1& 339 $\pm$ 257& 339 & 66049.00 &183 $\pm$ 65.4 &183   &  4277.16 \\
${\bf z}_8$ & -1 & -1 & -1 & -1 & -1 & -1 & -1 & 125 $\pm$ 146&125 & 21316.00 & 41.5 $\pm$ 45.1 &41.5 & 2034.01 \\
\bottomrule
\end{tabular*}%
\end{table}

Table \ref{exam_table1} presents the mean $\pm$ standard deviation of the results from five repeated observed outcomes. According to the preceding definitions of $\bar{Y}^{obs}({\bf z}_j)$ and $s^2({\bf z}_j)$, the mean reported in Example \ref{exam1} corresponds to $\bar{Y}^{obs}({\bf z}_j)$, and the square of its standard deviation corresponds to $s^2({\bf z}_j)$. $\hat\tau(k)$ is first calculated using (\ref{3.1}), followed by the calculation of $\widehat{\text{Var}}_{\text{N}}(\hat{\tau}(k))$ using (\ref{cor}). Finally, the confidence interval for Neymanian inference is determined by (\ref{NeyCI}). The final results are listed in Table \ref{exam_sum1} and Figure \ref{fig:example1}.

\begin{figure}[h]
\centering
\begin{tabular}{cc}
\includegraphics[width=0.5\textwidth]{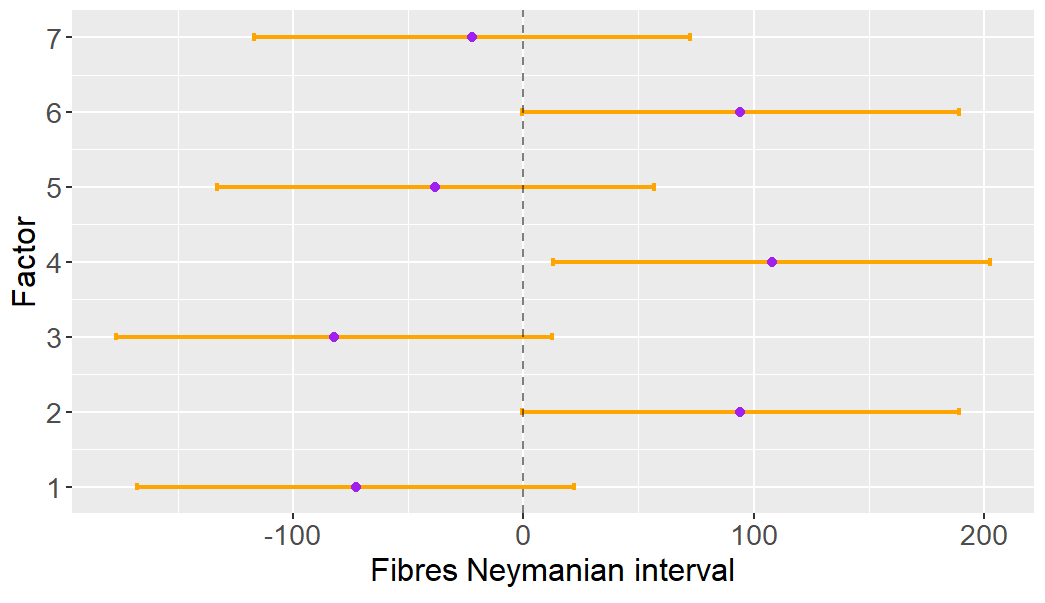} &
\includegraphics[width=0.5\textwidth]{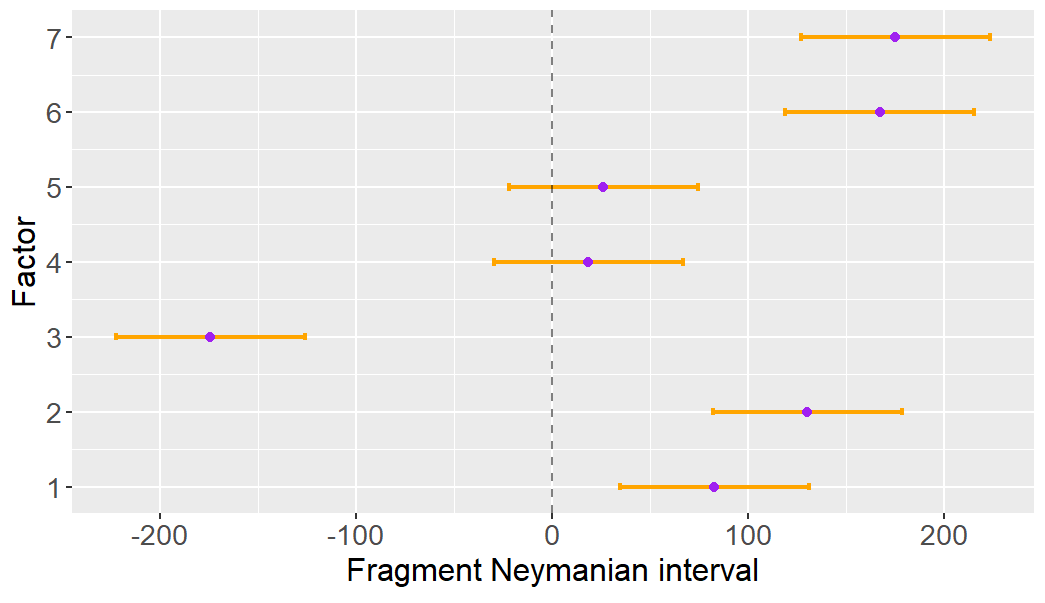}
\end{tabular}
\caption{95\% Neymanian confidence intervals for Example \ref{exam1}.} \label{fig:example1}
\end{figure}

In Neymanian inference within the potential outcomes framework, if the confidence interval for a factor's effect does not include 0, then the factor may be regarded as having a nonzero average causal effect on the corresponding outcome variable. As shown in Figure \ref{fig:example1}, for fibers, the confidence intervals for the corresponding average causal effects of factor 4 do not include 0. For fragments, the confidence intervals of the average causal effects of factors 1-3, 6, and 7 do not include 0. Therefore, it can be concluded that the average causal effects of factor 4 on fibers are nonzero, and those of factors 1-3, 6, and 7 on fragments are also nonzero.

\begin{table}[!http]
\caption{Summary of analysis of 95\% Neymainan intervals of PB design reported in Example \ref{exam1}.}
\setlength{\tabcolsep}{1pt}
\label{exam_sum1}
\begin{tabular*}{\linewidth}{@{\extracolsep{\fill}}{c}rrrr@{}}
\toprule
\multirow{2}{*}{Factor} 
& \multicolumn{2}{r}{Fibres} 
& \multicolumn{2}{r}{Fragments}  \\
\cmidrule(lr){2-3}
\cmidrule(lr){4-5}
 &Point estimate&Neymaian interval &Point estimate&Neymaian interval\\
\midrule
1 & -72.75 & $[-167.5, 22.03]$ & 82.70 & $[34.49, 130.9]$ \\
2 & 94.25 & $[-0.530, 189.0]$ & 130.20 & $[81.99, 178.4]$ \\
3 & -82.25 & $[-177.0, 12.53]$ & -174.45 & $[-222.7, -126.2]$ \\
4 & 107.75 & $[12.97, 202.5]$ & 18.55 & $[-29.66, 66.76]$ \\
5 & -38.25 & $[-133.0, 56.53]$ & 26.05 & $[-22.16, 74.26]$ \\
6 & 94.25 & $[-0.530, 189.0]$ & 167.20 & $[119.0, 215.4]$ \\
7 & -22.25 & $[-117.0, 72.53]$ & 175.05 & $[126.8, 223.3]$ \\
\bottomrule
\end{tabular*}%
\end{table}
Further examining Table \ref{exam_sum1}, for fibers, factor 4 has a positive effect, while for fragments, factor 3 has a negative effect, and factors 1, 2, 6, and 7 have positive effects. These directions are consistent with those identified in Example \ref{exam1}. Moreover, the factors identified by the Neymanian inference cover those selected in the original study based on whether the effect size exceeded twice the average standard deviation.

{\bf Example \ref{exam2}.} (Cont.) \cite{sahu2017screening} provided mean $\pm$ standard deviation of three repeated observed outcomes in Table \ref{exam2_or}.  The PB design with 12 runs is complete and non-geometric. The original paper is based on a statistical inference framework, they first examined the relative magnitudes of factor effects by ranking standardized effects in the Pareto Plot, then determined effect directions using regression coefficients and a 3D plot, thereby identifying the primary influential factors.

\begin{table}[!http]
\centering
\caption{Original design matrix and  data for Example \ref{exam2}.}
\label{exam2_or}
\tiny
\setlength{\tabcolsep}{1.2pt}
\renewcommand{\arraystretch}{1.08}
\resizebox{\textwidth}{!}{
\begin{tabular}{r*{11}{r}rrrrrrrr}
\toprule
\multirow{2}{*}{${\bf z}_{j}$}
& \multicolumn{11}{r}{Design factors}
& \multicolumn{4}{c}{Outcome variable column 1}
& \multicolumn{4}{c}{Outcome variable column 2} \\
\cmidrule(lr){2-12}
\cmidrule(lr){13-16}
\cmidrule(lr){17-20}
& $1$ & $2$ & $3$ & $4$ & $5$ & $6$
& $7$ & $8$ & $9$ & ${10}$ & ${11}$
& variable & mean $\pm$ SD & $\bar{Y}^{obs}({\bf z}_j)$ & $s^2({\bf z}_j)$
& variable & mean $\pm$ SD& $\bar{Y}^{obs}({\bf z}_j)$ & $s^2({\bf z}_j)$ \\
\midrule
\multirow{2}{*}{${\bf z}_{1}$}
& \multirow{2}{*}{$1$} & \multirow{2}{*}{$-1$} & \multirow{2}{*}{$-1$}
& \multirow{2}{*}{$-1$} & \multirow{2}{*}{$1$} & \multirow{2}{*}{$-1$}
& \multirow{2}{*}{$1$} & \multirow{2}{*}{$1$} & \multirow{2}{*}{$-1$}
& \multirow{2}{*}{$1$} & \multirow{2}{*}{$1$}
& VS & 278 $\pm$ 5.00  & 278 & 25 
& ZP & 20.25 $\pm$ 0.97 & 20.25 & 0.9409\\
\multicolumn{12}{c}{}
& EE & 58.45 $\pm$ 0.70 & 58.45 & 0.4900
& LC & 10.47 $\pm$ 0.62 & 10.47 & 0.3844  \\
\multirow{2}{*}{${\bf z}_{2}$}
& \multirow{2}{*}{$-1$}
& \multirow{2}{*}{$1$}
& \multirow{2}{*}{$1$}
& \multirow{2}{*}{$-1$}
& \multirow{2}{*}{$1$}
& \multirow{2}{*}{$1$}
& \multirow{2}{*}{$1$}
& \multirow{2}{*}{$-1$}
& \multirow{2}{*}{$-1$}
& \multirow{2}{*}{$-1$}
& \multirow{2}{*}{$1$}
& VS & 282 $\pm$ 8.72 & 282 & 76.0384
& ZP & 38.15 $\pm$ 1.67 & 38.15 & 2.7889 \\
\multicolumn{12}{c}{}
& EE & 73.40 $\pm$ 0.79 & 73.40 & 0.6241 
& LC & 3.54 $\pm$ 0.49 & 3.54 & 0.2401  \\
\multirow{2}{*}{${\bf z}_{3}$}
& \multirow{2}{*}{$-1$}
& \multirow{2}{*}{$1$}
& \multirow{2}{*}{$-1$}
& \multirow{2}{*}{$1$}
& \multirow{2}{*}{$1$}
& \multirow{2}{*}{$-1$}
& \multirow{2}{*}{$1$}
& \multirow{2}{*}{$1$}
& \multirow{2}{*}{$1$}
& \multirow{2}{*}{$-1$}
& \multirow{2}{*}{$-1$}
& VS & 682 $\pm$ 17.06 & 682 & 291.0436
& ZP &34.72 $\pm$ 1.27 & 34.72 & 1.6129\\
\multicolumn{12}{c}{}
& EE & 80.51 $\pm$ 1.31 & 80.51 & 1.7161   
& LC &  3.86 $\pm$ 0.58 & 3.86 & 0.3364\\
\multirow{2}{*}{${\bf z}_{4}$}
& \multirow{2}{*}{$-1$}
& \multirow{2}{*}{$1$}
& \multirow{2}{*}{$1$}
& \multirow{2}{*}{$1$}
& \multirow{2}{*}{$-1$}
& \multirow{2}{*}{$-1$}
& \multirow{2}{*}{$-1$}
& \multirow{2}{*}{$1$}
& \multirow{2}{*}{$-1$}
& \multirow{2}{*}{$1$}
& \multirow{2}{*}{$1$}
& VS & 292 $\pm$ 7.94 & 292 & 63.0436
& ZP & 39.20 $\pm$ 0.71 & 39.20 & 0.5041 \\
\multicolumn{12}{c}{}
& EE & 76.31 $\pm$ 0.84 & 76.31 & 0.7056 
& LC & 3.67 $\pm$ 0.49 & 3.67 & 0.2401\\
\multirow{2}{*}{${\bf z}_{5}$}
& \multirow{2}{*}{$1$}
& \multirow{2}{*}{$1$}
& \multirow{2}{*}{$1$}
& \multirow{2}{*}{$-1$}
& \multirow{2}{*}{$-1$}
& \multirow{2}{*}{$-1$}
& \multirow{2}{*}{$1$}
& \multirow{2}{*}{$-1$}
& \multirow{2}{*}{$1$}
& \multirow{2}{*}{$1$}
& \multirow{2}{*}{$-1$}
& VS & 990 $\pm$ 18.02 & 990 & 324.7204 
& ZP & 35.71 $\pm$ 1.04 & 35.71 & 1.0816 \\
\multicolumn{12}{c}{}
& EE & 81.89 $\pm$ 1.18 & 81.89 & 1.3924
& LC & 7.57 $\pm$ 0.48 & 7.57 & 0.2304 \\
\multirow{2}{*}{${\bf z}_{6}$}
& \multirow{2}{*}{$-1$}
& \multirow{2}{*}{$-1$}
& \multirow{2}{*}{$1$}
& \multirow{2}{*}{$-1$}
& \multirow{2}{*}{$1$}
& \multirow{2}{*}{$1$}
& \multirow{2}{*}{$-1$}
& \multirow{2}{*}{$1$}
& \multirow{2}{*}{$1$}
& \multirow{2}{*}{$1$}
& \multirow{2}{*}{$-1$}
& VS & 451 $\pm$ 9.16 & 451 & 83.9056
& ZP & 23.35 $\pm$ 0.70 & 23.35 & 0.4900 \\
\multicolumn{12}{c}{}
& EE & 51.41 $\pm$ 0.52 & 51.41 & 0.2704 
& LC & 4.89 $\pm$ 0.12 & 4.89 & 0.0144  \\
\multirow{2}{*}{${\bf z}_{7}$}
& \multirow{2}{*}{$-1$}
& \multirow{2}{*}{$-1$}
& \multirow{2}{*}{$-1$}
& \multirow{2}{*}{$-1$}
& \multirow{2}{*}{$-1$}
& \multirow{2}{*}{$-1$}
& \multirow{2}{*}{$-1$}
& \multirow{2}{*}{$-1$}
& \multirow{2}{*}{$-1$}
& \multirow{2}{*}{$-1$}
& \multirow{2}{*}{$-1$}
& VS & 322 $\pm$ 8.89 & 322 & 79.0321
& ZP & 19.39 $\pm$ 0.55 & 19.39 & 0.3025  \\
\multicolumn{12}{c}{}
& EE &  45.25 $\pm$ 0.75 & 45.25 & 0.5625
& LC & 4.33 $\pm$ 0.38 & 4.33 & 0.1444\\
\multirow{2}{*}{${\bf z}_{8}$}
& \multirow{2}{*}{$1$}
& \multirow{2}{*}{$1$}
& \multirow{2}{*}{$-1$}
& \multirow{2}{*}{$1$}
& \multirow{2}{*}{$1$}
& \multirow{2}{*}{$1$}
& \multirow{2}{*}{$-1$}
& \multirow{2}{*}{$-1$}
& \multirow{2}{*}{$-1$}
& \multirow{2}{*}{$1$}
& \multirow{2}{*}{$-1$}
& VS & 612 $\pm$ 7.82 & 612 & 61.1524
& ZP & 33.28 $\pm$ 0.87 & 33.28 & 0.7569  \\
\multicolumn{12}{c}{}
& EE &  83.20 $\pm$ 0.60 & 83.20 & 0.3600
& LC &7.68 $\pm$ 0.63 & 7.68 & 0.3969 \\
\multirow{2}{*}{${\bf z}_{9}$}
& \multirow{2}{*}{$-1$}
& \multirow{2}{*}{$-1$}
& \multirow{2}{*}{$-1$}
& \multirow{2}{*}{$1$}
& \multirow{2}{*}{$-1$}
& \multirow{2}{*}{$1$}
& \multirow{2}{*}{$1$}
& \multirow{2}{*}{$-1$}
& \multirow{2}{*}{$1$}
& \multirow{2}{*}{$1$}
& \multirow{2}{*}{$1$}
& VS & 112 $\pm$ 4.36 & 112 & 19.0096
& ZP & 20.79 $\pm$ 0.23 & 20.79 & 0.0529\\
\multicolumn{12}{c}{}
& EE & 55.61 $\pm$ 0.62 & 55.61 & 0.3844 
& LC & 5.27 $\pm$ 0.22 & 5.27 & 0.0484\\
\multirow{2}{*}{${\bf z}_{10}$}
& \multirow{2}{*}{$1$}
& \multirow{2}{*}{$-1$}
& \multirow{2}{*}{$1$}
& \multirow{2}{*}{$1$}
& \multirow{2}{*}{$-1$}
& \multirow{2}{*}{$1$}
& \multirow{2}{*}{$1$}
& \multirow{2}{*}{$1$}
& \multirow{2}{*}{$-1$}
& \multirow{2}{*}{$-1$}
& \multirow{2}{*}{$-1$}
& VS & 787 $\pm$ 12.12 & 787 & 146.8944
& ZP & 22.75 $\pm$ 1.28 & 22.75 & 1.6384 \\
\multicolumn{12}{c}{}
& EE & 56.11 $\pm$ 1.03 & 56.11 & 1.0609  
& LC & 10.09 $\pm$ 0.41 & 10.09 & 0.1681 \\
\multirow{2}{*}{${\bf z}_{11}$}
& \multirow{2}{*}{$1$}
& \multirow{2}{*}{$1$}
& \multirow{2}{*}{$-1$}
& \multirow{2}{*}{$-1$}
& \multirow{2}{*}{$-1$}
& \multirow{2}{*}{$1$}
& \multirow{2}{*}{$-1$}
& \multirow{2}{*}{$1$}
& \multirow{2}{*}{$1$}
& \multirow{2}{*}{$-1$}
& \multirow{2}{*}{$1$}
& VS & 258 $\pm$ 13.07 & 258 & 170.8249
& ZP & 35.71 $\pm$ 0.66 & 35.71 & 0.4356 \\
\multicolumn{12}{c}{}
& EE & 87.60 $\pm$ 0.60 & 87.60 & 0.0961
& LC & 8.05 $\pm$ 0.31 & 8.05 & 0.0961\\
\multirow{2}{*}{${\bf z}_{12}$}
& \multirow{2}{*}{$1$}
& \multirow{2}{*}{$-1$}
& \multirow{2}{*}{$1$}
& \multirow{2}{*}{$1$}
& \multirow{2}{*}{$1$}
& \multirow{2}{*}{$-1$}
& \multirow{2}{*}{$-1$}
& \multirow{2}{*}{$-1$}
& \multirow{2}{*}{$1$}
& \multirow{2}{*}{$-1$}
& \multirow{2}{*}{$1$}
& VS & 211 $\pm$ 11.36 & 211 & 129.0496
& ZP & 25.35 $\pm$ 0.84 & 25.35 & 0.7056 \\
\multicolumn{12}{c}{}
& EE & 61.71 $\pm$ 0.83 & 61.71 & 0.6889
& LC & 11.00 $\pm$ 0.50 & 11.00 & 0.2500 \\
\bottomrule
\end{tabular}
}
\end{table}

As shown in Table \ref{exam2_or}, according to the preceding definitions of $\bar{Y}^{obs}({\bf z}_j)$ and $s^2({\bf z}_j)$, The mean reported in Example \ref{exam2} corresponds to $\bar{Y}^{obs}({\bf z}_j)$, and the square of its standard deviation corresponds to $s^2({\bf z}_j)$.
The point estimate can be computed using (\ref{3.1}), while the confidence interval for Neymanian inference can be obtained from (\ref{NeyCI}). The final results are listed in Table \ref{exam_sum2}, and the corresponding plot is shown in Figure \ref{fig:exa 2}.

\begin{figure}[h]
\centering
\begin{tabular}{cc}
\includegraphics[width=0.5\textwidth]{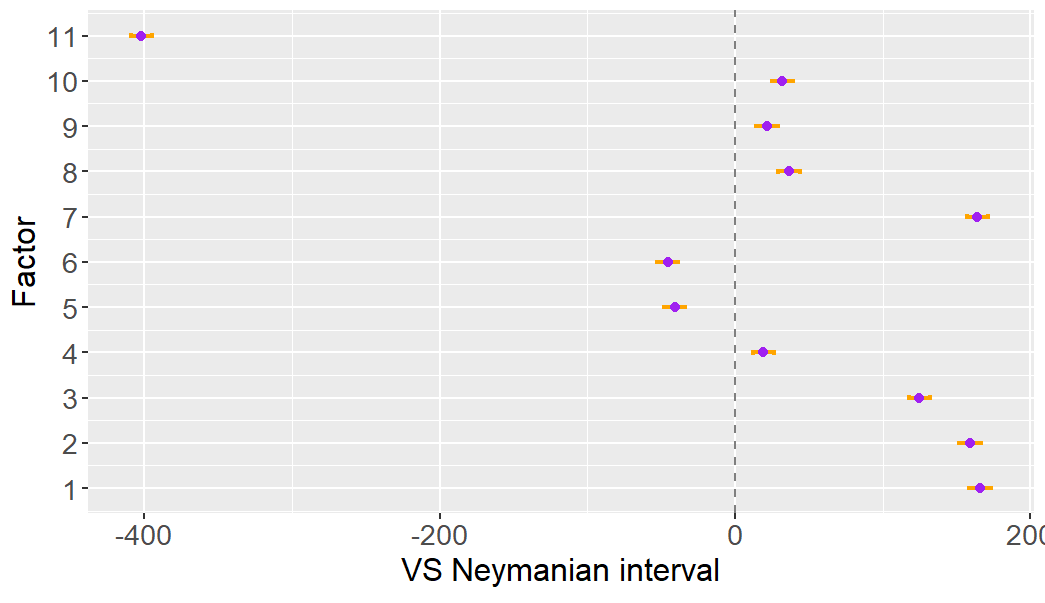} &
\includegraphics[width=0.5\textwidth]{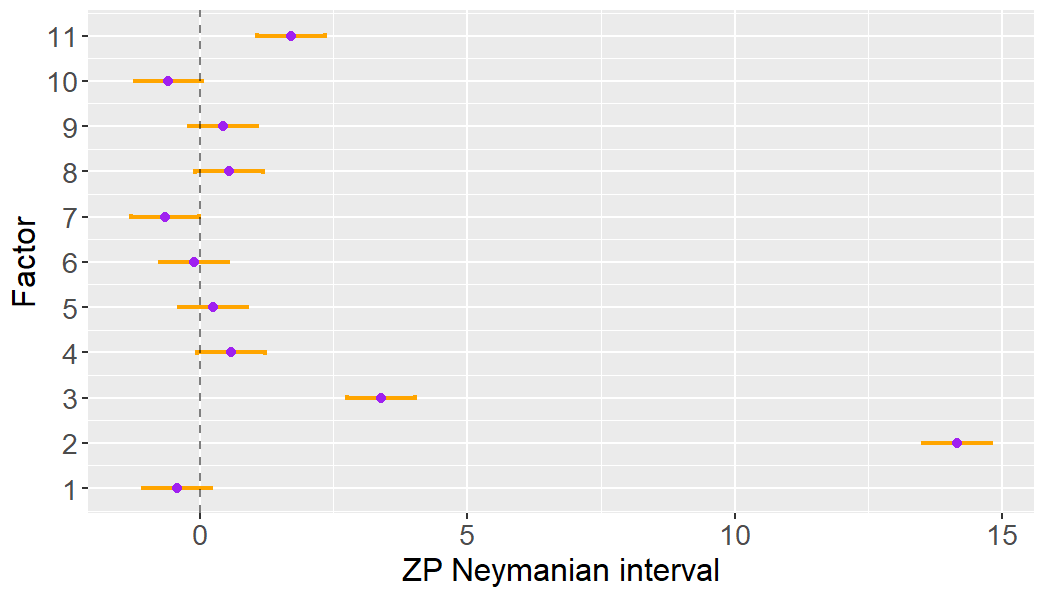} \\[0.5em]
\includegraphics[width=0.5\textwidth]{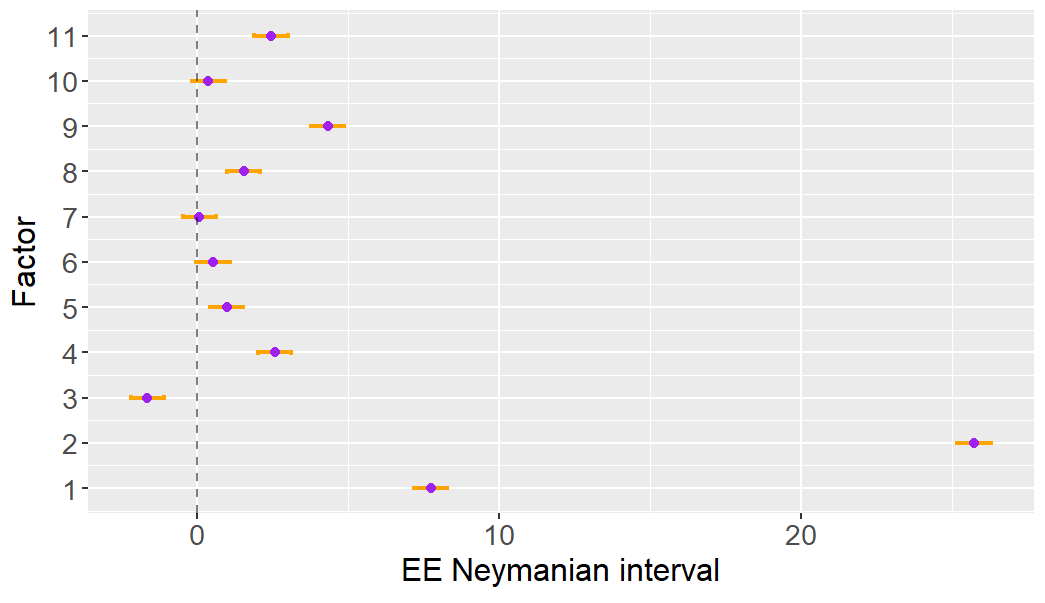} &
\includegraphics[width=0.5\textwidth]{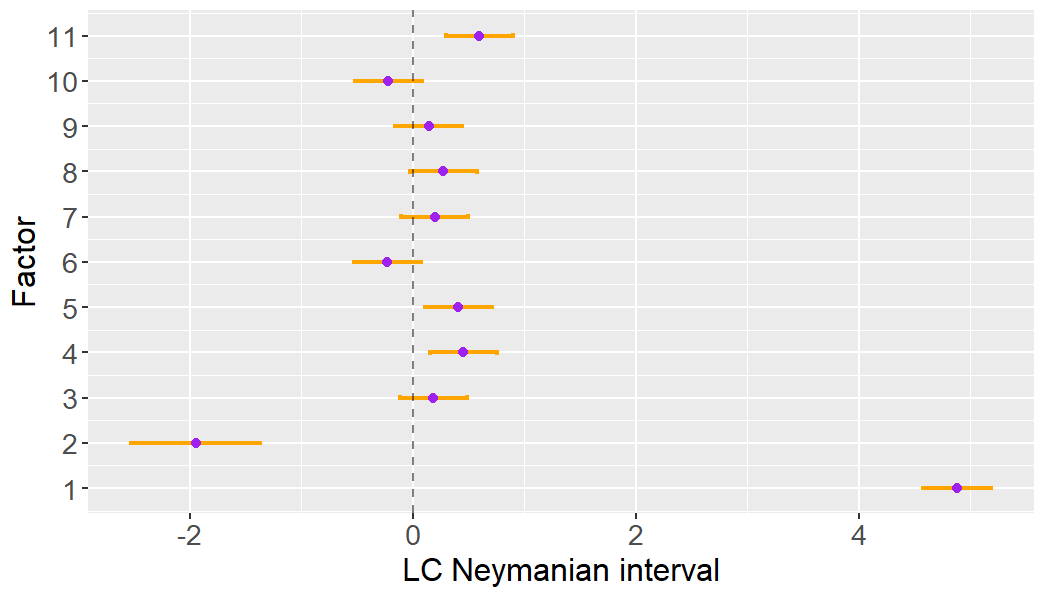}
\end{tabular}
\caption{95\% Neymanian confidence intervals for Example \ref{exam2}.} \label{fig:exa 2}
\end{figure}

\begin{table}[!http]
\centering
\caption{Summary of analysis of 95\% Neymainan intervals of PB design reported in Example \ref{exam2}.}
\footnotesize
\setlength{\tabcolsep}{0.1pt}
\label{exam_sum2}
\begin{tabular*}{\linewidth}{@{\extracolsep{\fill}}{c}crcrcrcr@{}}
\toprule
\multirow{2}{*}{Factor} 
& \multicolumn{2}{c}{VS} 
& \multicolumn{2}{c}{ZP} 
& \multicolumn{2}{c}{EE}
& \multicolumn{2}{c}{LC}\\
\cmidrule(lr){2-3}
\cmidrule(lr){4-5}
\cmidrule(lr){6-7}
\cmidrule(lr){8-9}
&\makecell{Point\\estimate} & \makecell{Neymaian\\interval} 
& \makecell{Point\\estimate} & \makecell{Neymaian\\interval} 
& \makecell{Point\\estimate} & \makecell{Neymaian\\interval} 
& \makecell{Point\\estimate} & \makecell{Neymaian\\interval} \\
\hline
1   &165.833 &$[158.6 ,173.1]$ &-0.425&$[-1.059 ,0.209 ]$  &7.745
 &$[7.191,8.299 ]$ &4.883 &$[4.582,5.184]$  \\
2   &159.167 & $[151.9,166.4 ]$ &14.148 & $[13.51,14.78]$ &25.728 & $[25.18,26.28]$ &-1.947 & $[-2.248, -1.646]$\\
3  & 124.833  & $[117.6,132.1 ]$ &3.395& $[2.761,4.029]$ &-1.632 &  $[-2.185 ,-1.078 ]$&0.183 & $[-0.118,0.484]$  \\
4  &19.167 &$[11.97,26.40]$ &0.588 &$[-0.046,1.223]$    &2.575 &$[2.021,3.129]$ &0.453  & $[0.152 ,0.754]$  \\
5  &-40.833 & $[-48.06,-33.60]$  &0.258  & $[-0.376,0.893]$  &0.985  & $[0.431 , 1.539  ]$  &0.410 & $[0.109,0.711 ]$  \\
6  & -45.5 & $[-52.73, -38.27 ]$ &-0.098 & $[-0.733,0.536 ]$ & 0.535 & $[-0.019,1.089]$ &-0.230 &$[-0.531,0.071]$  \\
7 &164.167  &$[156.9,171.4]$  &-0.652 &$[-1.286,-0.017 ]$  &0.082&$[-0.472,0.635]$  &0.197 & $[-0.104,0.498 ]$  \\
8   &36.50 & $[29.27,43.73]$ &0.552 & $[-0.083,1.186 ]$  &1.555& $[1.001,2.109]$ &0.273& $[-0.028 , 0.574]$  \\
9   & 21.833 & $[14.60,29.06 ]$  &0.435 & $[-0.199 ,1.069 ]$  & 4.335 &$$[3.781 ,4.889]$$ &0.143 &$[-0.158,0.444]$  \\
10 &32.167 &$[24.94,39.40]$ &-0.582 &$[-1.216 ,0.053 ]$  &0.382 
&$[-0.172 ,0.935 ]$ &-0.22 & $[-0.521,0.081]$  \\
11  &-401.833 & $[-409.0,-394.60]$  &1.708 & $[1.074,2.343]$ &2.452 & $[1.898,3.005]$  &0.597 
 & $[0.296,0.898]$  \\
\hline
\end{tabular*}
\end{table}
As shown in Figure \ref{fig:exa 2} and Table \ref{exam_sum2}, for the outcome variable VS, the causal effects of all factors are nonzero; for ZP, the causal effects of factors 2, 3, 7, and 11 are nonzero; for EE, the causal effects of factors 1-5, 8, 9 and 11 are nonzero; and for LC, the causal effects of factors 1, 2, 4, 5, and 11 are nonzero. For each of the four outcome variables, the three factors with the largest absolute average causal effects were selected. Specifically, the selected factors were 11, 1, and 7 for VS; 2, 3, and 11 for ZP; 1, 2, and 9 for EE; and 1, 2, and 11 for LC. 

To satisfy the requirements that VS be as small as possible and EE be as large as possible, we selected the candidate factors for the next stage by considering both the directions and magnitudes of the effects reported in Table 1. For VS, the three factors with the largest absolute causal effects were 11, 1, and 7; among them, only 11 had a negative effect, consistent with the objective of minimizing VS, and was therefore retained for the next stage. Its effect value was -401.83, indicating that when factor 11 changes from the low level to the high level, the mean response of VS decreases by 401.83. For ZP, the effects of factors 2, 3, and 11 were all positive, so all three were retained. For EE and LC, since only the effect magnitude was considered, the three factors with the largest absolute causal effects were selected for each outcome variable. Finally, taking the union of the selected factors across all outcome variables yielded factors 1, 2, 3, 9, and 11 as candidate factors for the next stage, consistent with the original study's conclusion.

\begin{table}[h]
\centering
\caption{Original design matrix and  data for Example \ref{exam3}.}
\label{exam3_or}
\tiny
\setlength{\tabcolsep}{1.2pt}
\renewcommand{\arraystretch}{1.08}
\resizebox{\textwidth}{!}{
\begin{tabular}{r*{11}{r}rrrrrrrr}
\toprule
\multirow{2}{*}{${\bf z}_{j}$}
& \multicolumn{8}{c}{Design factors}
& \multicolumn{4}{c}{Outcome variable column 1}
& \multicolumn{4}{c}{Outcome variable column 2} \\
\cmidrule(lr){2-9}
\cmidrule(lr){10-13}
\cmidrule(lr){14-17}
& $1$ & $2$ & $3$ & $4$ & $5$ & $6$
& $7$ & $8$ 
& variable & mean $\pm$ SD & $\bar{Y}^{obs}({\bf z}_j)$ & $s^2({\bf z}_j)$
& variable & mean $\pm$ SD& $\bar{Y}^{obs}({\bf z}_j)$ & $s^2({\bf z}_j)$ \\
\midrule
\multirow{2}{*}{${\bf z}_{1}$}
& \multirow{2}{*}{$1$} & 
\multirow{2}{*}{$1$} & 
\multirow{2}{*}{$-1$} & 
\multirow{2}{*}{$1$} & 
\multirow{2}{*}{$1$} & 
\multirow{2}{*}{$1$} & 
\multirow{2}{*}{$-1$} & 
\multirow{2}{*}{$-1$}
& MPS & $411.5 \pm 3.1$ & 411.5 &  9.61
& ZP & $-3.15 \pm 0.07$ &  -3.15 &$4.9\times10^{-3}$\\
\multicolumn{9}{c}{}
& EE & $65.2 \pm 2.6$ & 65.2 & 6.76 
& $DE_{30}$ & $25.2 \pm 0.8$ & 25.2 & 0.64   \\
\multirow{2}{*}{${\bf z}_{2}$}&
\multirow{2}{*}{$-1$} & 
\multirow{2}{*}{$1$} & 
\multirow{2}{*}{$1$} & 
\multirow{2}{*}{$-1$} & 
\multirow{2}{*}{$1$} & 
\multirow{2}{*}{$1$} & 
\multirow{2}{*}{$1$} & 
\multirow{2}{*}{$-1$} 
& MPS & $531.4 \pm 4.6$ & 531.4 & 21.16
& ZP & $-11.20 \pm 0.09$ & -11.20 & $8.1\times10^{-3}$ \\
\multicolumn{9}{c}{}
& EE & $67.3 \pm 1.9$ & 67.3 & 3.61 
& $DE_{30}$ & $30.9 \pm 1.5$ & 30.9 & 2.25 \\
\multirow{2}{*}{${\bf z}_{3}$}
& \multirow{2}{*}{$1$} & 
\multirow{2}{*}{$-1$} & 
\multirow{2}{*}{$1$} & 
\multirow{2}{*}{$1$} & 
\multirow{2}{*}{$-1$} & 
\multirow{2}{*}{$1$} & 
\multirow{2}{*}{$1$} & 
\multirow{2}{*}{$1$} 
& MPS & $477.5 \pm 3.6$ & 477.5 & 12.96
& ZP & $-13.10 \pm 0.22$ & -13.10 &$4.84\times10^{-2}$\\
\multicolumn{9}{c}{}
& EE &$47.2 \pm 1.5$ & 47.2 & 2.25    
& $DE_{30}$ &  $77.8 \pm 2.8$ & 77.8 & 7.84\\
\multirow{2}{*}{${\bf z}_{4}$}
& \multirow{2}{*}{$-1$} &
\multirow{2}{*}{$1$} &
\multirow{2}{*}{$-1$} &
\multirow{2}{*}{$1$} &
\multirow{2}{*}{$1$} &
\multirow{2}{*}{$-1$} &
\multirow{2}{*}{$1$} &
\multirow{2}{*}{$1$} 
& MPS &$278.9 \pm 2.2$ & 278.9 &  4.84
& ZP &$-2.35 \pm 0.06$ &  -2.35 &$3.6\times10^{-3}$\\
\multicolumn{9}{c}{}
& EE & $67.3 \pm 2.3$ & 67.3 & 5.29
& $DE_{30}$ & $31.3 \pm 1.1$ & 31.3 & 1.21 \\
\multirow{2}{*}{${\bf z}_{5}$}
& \multirow{2}{*}{$-1$} &
\multirow{2}{*}{$-1$} &
\multirow{2}{*}{$1$} &
\multirow{2}{*}{$-1$} &
\multirow{2}{*}{$1$} &
\multirow{2}{*}{$1$} &
\multirow{2}{*}{$-1$} &
\multirow{2}{*}{$1$} 
& MPS &$552.6 \pm 3.7$ & 552.6 & 13.69
& ZP & $-10.50 \pm 0.09$ & -10.50 & $8.1\times10^{-3}$\\
\multicolumn{9}{c}{}
& EE & $58.3 \pm 2.2$ & 58.3 & 4.84
& $DE_{30}$ &$73.4 \pm 1.4$ & 73.4 & 1.96 \\
\multirow{2}{*}{${\bf z}_{6}$}
& \multirow{2}{*}{$-1$} &
\multirow{2}{*}{$-1$} &
\multirow{2}{*}{$-1$} &
\multirow{2}{*}{$1$} &
\multirow{2}{*}{$-1$} &
\multirow{2}{*}{$1$} &
\multirow{2}{*}{$1$} &
\multirow{2}{*}{$-1$} 
& MPS &$387.7 \pm 2.6$ & 387.7 &  6.76
& ZP & $-3.82 \pm 0.03$ &  -3.82 &$9\times10^{-4}$\\
\multicolumn{9}{c}{}
& EE & $84.2 \pm 1.8$ & 84.2 & 3.24
& $DE_{30}$ & $85.1 \pm 1.3$ & 85.1 & 1.69 \\
\multirow{2}{*}{${\bf z}_{7}$}
& \multirow{2}{*}{$1$} &
\multirow{2}{*}{$-1$} &
\multirow{2}{*}{$-1$} &
\multirow{2}{*}{$-1$} &
\multirow{2}{*}{$1$} &
\multirow{2}{*}{$-1$} &
\multirow{2}{*}{$1$} &
\multirow{2}{*}{$1$}
& MPS & $422.3 \pm 3.7$ & 422.3 & 13.69
& ZP & $-1.04 \pm 0.02$ &  -1.04 &$4\times10^{-4}$ \\
\multicolumn{9}{c}{}
& EE & $82.1 \pm 1.7$ & 82.1 & 2.89
& $DE_{30}$  & $61.3 \pm 2.2$ & 61.3 & 4.84\\
\multirow{2}{*}{${\bf z}_{8}$}
& \multirow{2}{*}{$1$} &
\multirow{2}{*}{$1$} &
\multirow{2}{*}{$-1$} &
\multirow{2}{*}{$-1$} &
\multirow{2}{*}{$-1$} &
\multirow{2}{*}{$1$} &
\multirow{2}{*}{$-1$} &
\multirow{2}{*}{$1$} 
& MPS & $591.2 \pm 4.7$ & 591.2 & 22.09
& ZP & $-3.32 \pm 0.04$ &  -3.32 &$1.6\times10^{-3}$\\
\multicolumn{9}{c}{}
& EE & $46.3 \pm 2.1$ & 46.3 & 4.41
& $DE_{30}$ &$41.1 \pm 1.7$ & 41.1 & 2.89 \\
\multirow{2}{*}{${\bf z}_{9}$}
& \multirow{2}{*}{$1$} &
\multirow{2}{*}{$1$} &
\multirow{2}{*}{$1$} &
\multirow{2}{*}{$-1$} &
\multirow{2}{*}{$-1$} &
\multirow{2}{*}{$-1$} &
\multirow{2}{*}{$1$} &
\multirow{2}{*}{$-1$} 
& MPS &$514.2 \pm 4.4$ & 514.2 & 19.36
& ZP &$-12.20 \pm 0.19$ & -12.20 &$3.61\times10^{-2}$\\
\multicolumn{9}{c}{}
& EE &$29.9 \pm 1.5$ & 29.9 & 2.25
& $DE_{30}$ &$42.3 \pm 1.6$ & 42.3 & 2.56\\
\multirow{2}{*}{${\bf z}_{10}$}
& \multirow{2}{*}{$-1$} &
\multirow{2}{*}{$1$} &
\multirow{2}{*}{$1$} &
\multirow{2}{*}{$1$} &
\multirow{2}{*}{$-1$} &
\multirow{2}{*}{$-1$} &
\multirow{2}{*}{$-1$} &
\multirow{2}{*}{$1$} 
& MPS & $379.8 \pm 2.6$ & 379.8 &  6.76
& ZP & $-10.90 \pm 0.13$ & -10.90 &$1.69\times10^{-2}$ \\
\multicolumn{9}{c}{}
& EE & $28.3 \pm 1.5$ & 28.3 & 2.25   
& $DE_{30}$ & $47.9 \pm 2.1$ & 47.9 & 4.41 \\
\multirow{2}{*}{${\bf z}_{11}$}
& \multirow{2}{*}{$1$} &
\multirow{2}{*}{$-1$} &
\multirow{2}{*}{$1$} &
\multirow{2}{*}{$1$} &
\multirow{2}{*}{$1$} &
\multirow{2}{*}{$-1$} &
\multirow{2}{*}{$-1$} &
\multirow{2}{*}{$-1$} 
& MPS & $427.2 \pm 3.3$ & 427.2 & 10.89
& ZP & $-14.30 \pm 0.22$ & -14.30 &$4.84\times10^{-2}$  \\
\multicolumn{9}{c}{}
& EE & $59.9 \pm 2.4$ & 59.9 & 5.76 
& $DE_{30}$ &$70.6 \pm 1.5$ & 70.6 & 2.25\\
\multirow{2}{*}{${\bf z}_{12}$}
& \multirow{2}{*}{$-1$} &
\multirow{2}{*}{$-1$} &
\multirow{2}{*}{$-1$} &
\multirow{2}{*}{$-1$} &
\multirow{2}{*}{$-1$} &
\multirow{2}{*}{$-1$} &
\multirow{2}{*}{$-1$} &
\multirow{2}{*}{$-1$} 
& MPS & $398.9 \pm 2.6$ & 398.9 &  6.76 
& ZP & $-1.24 \pm 0.02$ &  -1.24 &$4\times10^{-4}$ \\
\multicolumn{9}{c}{}
& EE & $70.1 \pm 1.9$ & 70.1 & 3.61 
& $DE_{30}$ & $86.2 \pm 1.2$ & 86.2 & 1.44 \\
\bottomrule
\end{tabular}
}
\end{table}

{\bf Example \ref{exam3}.} (Cont.)  \cite{dhat2017risk} provided $\pm$ standard deviation of outcomes in Table \ref{exam3_or}. The PB design with 12 runs is incomplete and non-geometric. Unlike the previous two examples, after identifying the main factors through the PB design, they further examined their interactions using a $2^3$ full factorial design. The original paper is based on a statistical inference framework, a PB design is first employed, along with regression analysis and ANOVA, to identify the significant factor, and Pareto charts are used to identify the key factors with relatively large contributions. Subsequently, based on the screening results, a $2^3$ full factorial design was applied, and regression models, ANOVA, and Lenth's method were used to further analyze the significance of the main factors and their interactions. 

Based on Table \ref{exam3_or}, according to the preceding definitions of $\bar{Y}^{obs}({\bf z}_j)$ and $s^2({\bf z}_j)$, the mean reported in the original article corresponds to $\bar{Y}^{obs}({\bf z}_j)$, and the square of its standard deviation corresponds to $s^2({\bf z}_j)$. The point estimate can be computed using (\ref{3.1}), while the confidence interval for Neymanian inference can be obtained from (\ref{NeyCI}). The final results are listed in Table \ref{exam3_sum}, and the corresponding plot is shown in Figure \ref{fig:Example 2}.

\begin{figure}[h]
\centering
\begin{tabular}{cc}
\includegraphics[width=0.5\textwidth]{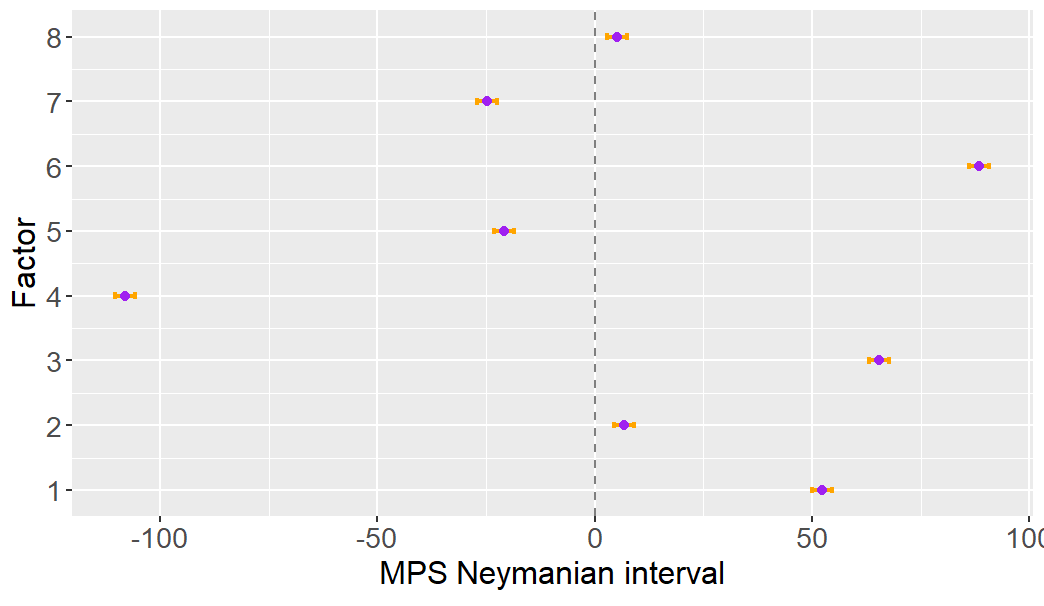} &
\includegraphics[width=0.5\textwidth]{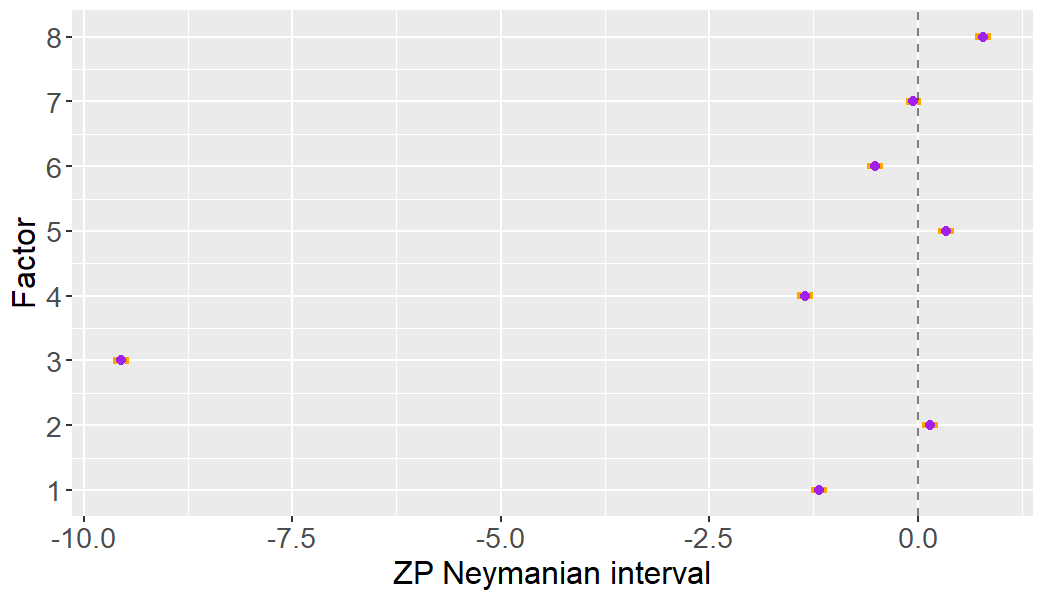} \\[0.5em]
\includegraphics[width=0.5\textwidth]{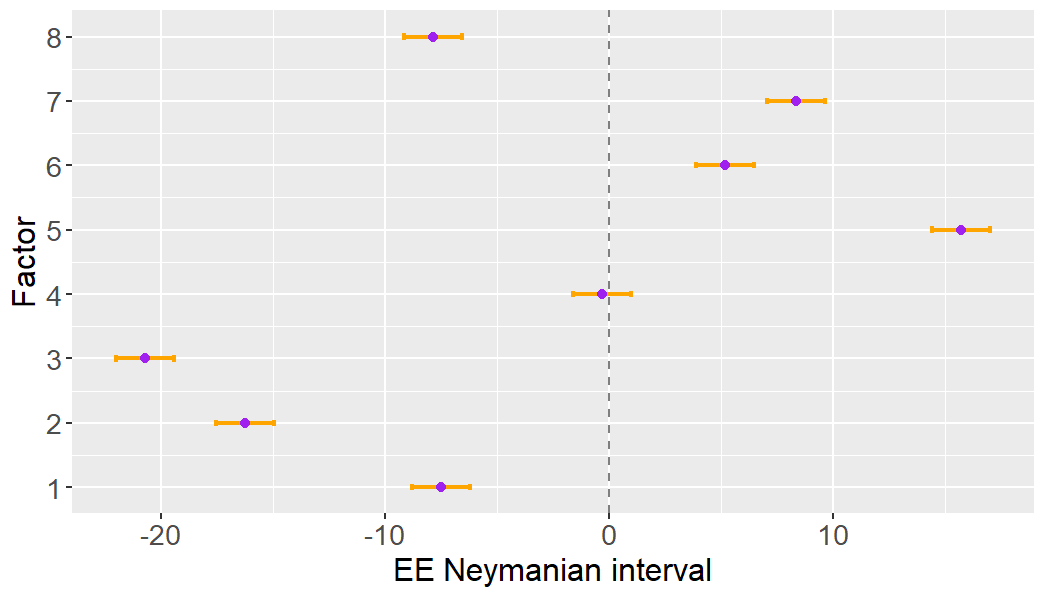} &
\includegraphics[width=0.5\textwidth]{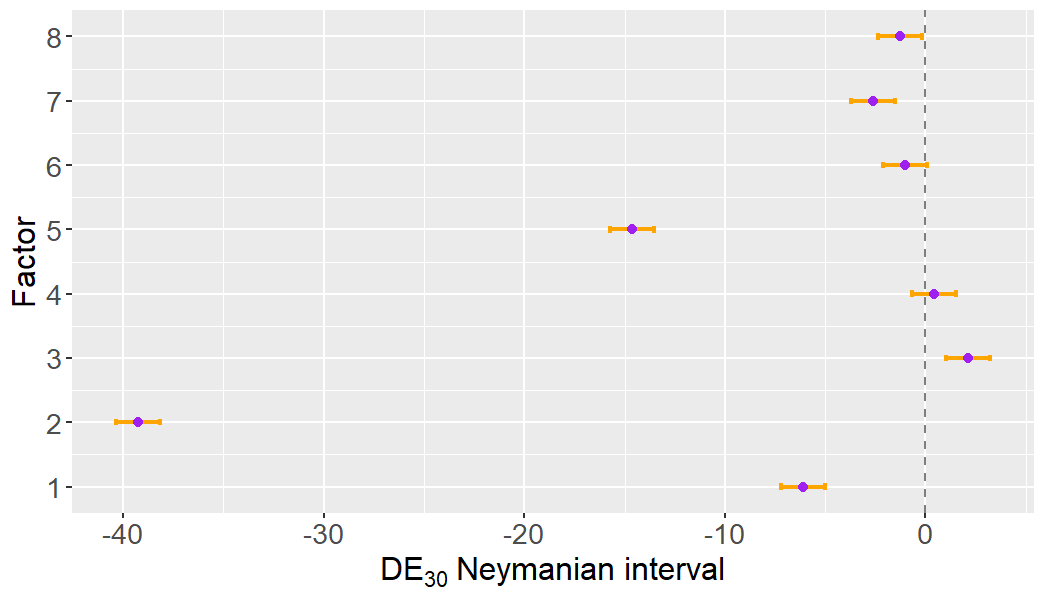}
\end{tabular}
\caption{95\% Neymanian confidence intervals of PB design for Example \ref{exam3}.} \label{fig:Example 2}
\end{figure}

\begin{table}[!http]
\caption{Summary of analysis of 95\% Neymainan intervals of PB design reported in Example \ref{exam3}.}
\footnotesize
\setlength{\tabcolsep}{0.1pt}
\label{exam3_sum}
\begin{tabular*}{\linewidth}{@{\extracolsep{\fill}}{c}crcrcrcr@{}}
\toprule
\multirow{2}{*}{Factor} 
& \multicolumn{2}{c}{MPS} 
& \multicolumn{2}{c}{ZP} 
& \multicolumn{2}{c}{EE}
& \multicolumn{2}{c}{$DE_{30}$}\\
\cmidrule(lr){2-3}
\cmidrule(lr){4-5}
\cmidrule(lr){6-7}
\cmidrule(lr){8-9}
&\makecell{Point\\estimate} & \makecell{Neymaian\\interval} 
& \makecell{Point\\estimate} & \makecell{Neymaian\\interval} 
& \makecell{Point\\estimate} & \makecell{Neymaian\\interval} 
& \makecell{Point\\estimate} & \makecell{Neymaian\\interval} \\
\hline
1 & 52.433 & [50.13, 54.73] & -1.183 & [-1.254, -1.112] & -7.483 & [-8.779, -6.188] & -6.083 & [-7.183, -4.984] \\
2 & 6.800  & [4.501, 9.099]   & 0.147  & [0.076, 0.218]   & -16.250 & [-17.55, -14.96] & -39.283 & [-40.38, -38.18] \\
3 & 65.367 & [63.07, 67.67] & -9.547 & [-9.618, -9.476] & -20.717 & [-22.01, -19.42] & 2.117 & [1.017, 3.216] \\
4 & -108.000 & [-110.3, -105.7] & -1.353 & [-1.424, -1.282] & -0.317 & [-1.612, 0.979] & 0.450 & [-0.649, 1.549] \\
5 & -20.900 & [-23.20, -18.60] & 0.340 & [0.269, 0.411] & 15.683 & [14.39, 16.98] & -14.617 & [-15.72, -13.52] \\
6 & 88.433 & [86.13, 90.73] & -0.510 & [-0.581, -0.439] & 5.150 & [3.855, 6.445] & -1.017 & [-2.116, 0.083] \\
7 & -24.867 & [-27.17, -22.57] & -0.050 & [-0.121, 0.021] & 8.317 & [7.021, 9.612] & -2.617 & [-3.716, -1.517] \\
8 & 5.233 & [2.934, 7.532] & 0.783 & [0.712, 0.854] & -7.850 & [-9.145, -6.555] & -1.250 & [-2.349, -0.151] \\
\hline
\end{tabular*}%
\end{table}

As shown in Figure \ref{fig:Example 2} and Table \ref{exam3_sum}, for the outcome variable MPS, the causal effects of all factors are nonzero; for ZP, the causal effects of factors 1-6, and 8 are nonzero; for EE, the causal effects of factors 1-3 and 4-8 are nonzero; and for $DE_{30}$, the causal effects of factors 1-3, 5,7 and 8 are nonzero. For each of the four outcome variables, the three factors with the largest absolute average causal effects were selected. Specifically, the selected factors were 4, 6, and 3 for MPS; 3, 4, and 1 for ZP; 3, 2, and 5 for EE; and 2, 5, and 1 for $DE_{30}$.

Next, by setting the objective to minimize MPS and maximize EE, we selected the factors to include in the next optimization stage based on the effect directions. Since factors 2 and 3 represent material types, and their effects on MPS are positive while their average causal effects on EE are negative, they were set to the low levels ES100 and PVA, respectively, to meet the preparation requirements. For factor 1, since changing it from the low level to the high level would increase MPS, it was also fixed at the low level. After excluding factors 1, 2, and 3 and combining the previously selected candidate factors based on the magnitudes of their causal effects, factors 4, 5, and 6 were included in the next optimization stage. In addition, factors 7 and 8 were fixed at the low levels for cost-saving considerations. The resulting selection of factors 4, 5, and 6 for the next stage is consistent with that reported in the original study.

In the original paper, these three factors were redefined as the PVA concentration in the aqueous phase, the volume of the aqueous phase, and the amount of ES100 in the organic phase. The three response variables under consideration were MPS, EE, and PDI. A $2^3$ full factorial design was then employed to investigate further the effects of these three factors as well as their interactions. The original data for the $2^3$ full factorial design are presented in Table \ref{tab:$2^3$ factorial design}. 
\begin{table}[h]
\centering
\caption{Original $2^3$ factorial design matrix and raw data for Example \ref{exam3}.}
\label{tab:$2^3$ factorial design}
\footnotesize
\setlength{\tabcolsep}{3pt}
\renewcommand{\arraystretch}{0.95}
\begin{tabular*}{\textwidth}{@{\extracolsep{\fill}}c*{3}{r}*{3}{r}*{3}{r}*{3}{r}@{}}
\toprule
\multirow{2}{*}{$\bf{z}_j$}
& \multicolumn{3}{c}{Factor}
& \multicolumn{3}{c}{MPS}
& \multicolumn{3}{c}{EE}
& \multicolumn{3}{c}{PDI} \\
\cmidrule(lr){2-4}
\cmidrule(lr){5-7}
\cmidrule(lr){8-10}
\cmidrule(lr){11-13}
& 1 & 2 & 3& mean $\pm$ SD
& $\bar{Y}^{obs}({\bf z}_j)$ & $s^2({\bf z}_j)$& mean $\pm$ SD
&  $\bar{Y}^{obs}({\bf z}_j)$ & $s^2({\bf z}_j)$& mean $\pm$ SD
&  $\bar{Y}^{obs}({\bf z}_j)$ & $s^2({\bf z}_j)$ \\
\midrule
${\bf z}_1$ & -1 & -1 & -1 & $420.3 \pm 3.6$ & 420.3 & 12.96 & $69.2 \pm 2.4$ & 69.2 & 5.76 & $0.32 \pm 0.010$ & 0.32 &$1\times10^{-4}$\\
${\bf z}_2$ &  1 & -1 & -1 & $210.6 \pm 2.1$ & 210.6 &  4.41 & $76.5 \pm 2.2$ & 76.5 & 4.84 & $0.22 \pm 0.021$ & 0.22 &$4.41\times10^{-4}$ \\
${\bf z}_3$ & -1 &  1 & -1 & $430.6 \pm 3.2$ & 430.6 & 10.24 & $65.5 \pm 2.5$ & 65.5 & 6.25 & $0.33 \pm 0.013$ & 0.33 &$1.69\times10^{-4}$\\
${\bf z}_4$ &  1 &  1 & -1 & $269.2 \pm 2.7$ & 269.2 &  7.29 & $74.3 \pm 2.9$ & 74.3 & 8.41 & $0.29 \pm 0.014$ & 0.29 &$1.96\times10^{-4}$ \\
${\bf z}_5$ & -1 & -1 &  1 & $450.2 \pm 3.1$ & 450.2 &  9.61 & $72.3 \pm 2.4$ & 72.3 & 5.76 & $0.32 \pm 0.011$ & 0.32 &$1.21\times10^{-4}$\\
${\bf z}_6$ &  1 & -1 &  1 & $305.5 \pm 2.9$ & 305.5 &  8.41 & $84.4 \pm 1.3$ & 84.4 & 1.69 & $0.34 \pm 0.013$ & 0.34 &$1.69\times10^{-4}$\\
${\bf z}_7$ & -1 &  1 &  1 & $465.7 \pm 3.8$ & 465.7 & 14.44 & $70.3 \pm 2.7$ & 70.3 & 7.29 & $0.33 \pm 0.009$ & 0.33 &$8.1\times10^{-5}$\\
${\bf z}_8$ &  1 &  1 &  1 & $233.6 \pm 2.4$ & 233.6 &  5.76 & $82.0 \pm 1.2$ & 82.0 & 1.44 & $0.32 \pm 0.007$ & 0.32 &$4.9 \times 10^{-5}$\\
\bottomrule
\end{tabular*}
\end{table}

Since the $2^3$ full factorial design is a two-level orthogonal array, according to Remark \ref{rem}, Neymanian confidence intervals can also be constructed for the $2^3$ full factorial design using (\ref{NeyCI}). The results are presented in Table \ref{exam32^3_sum}, and their graphical display is provided in Figure \ref{fig:Example 2}.

\begin{figure}[h]
\centering
\begin{tabular}{cc}
\includegraphics[width=0.45\textwidth]{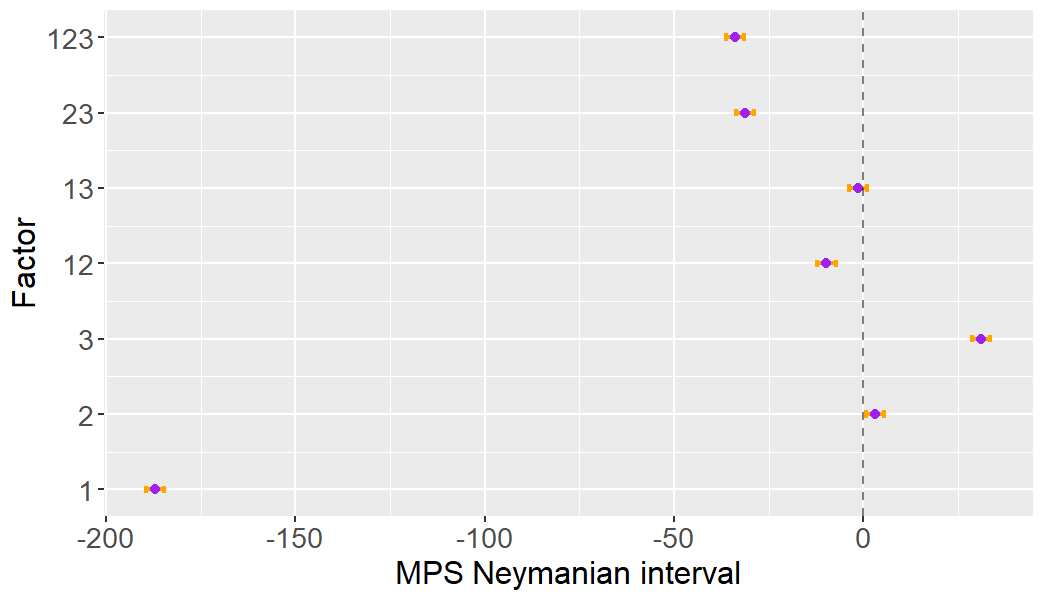} &
\includegraphics[width=0.45\textwidth]{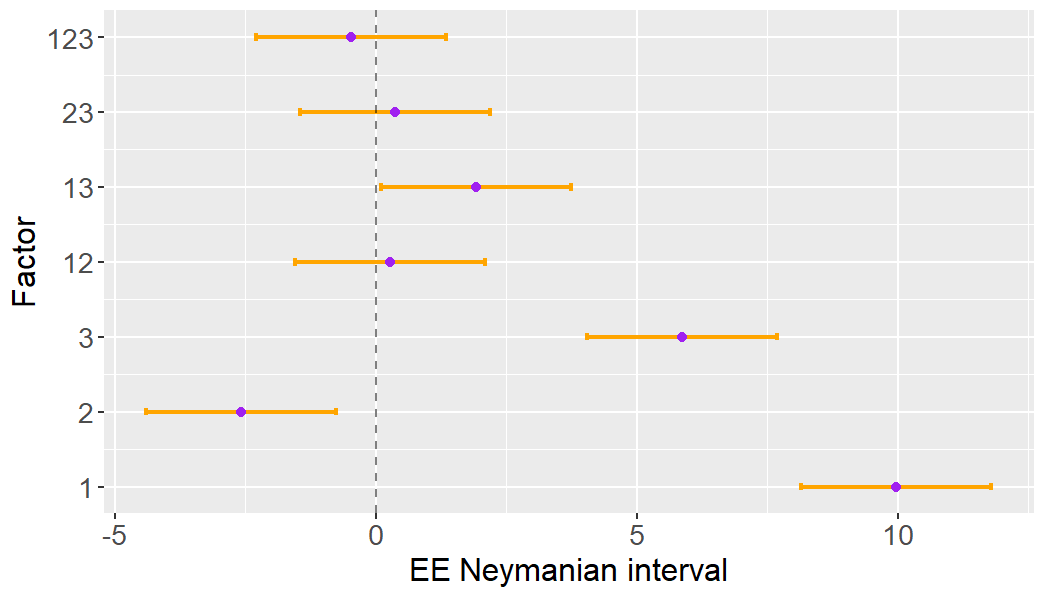} \\[2mm]
\multicolumn{2}{c}{
\includegraphics[width=0.45\textwidth]{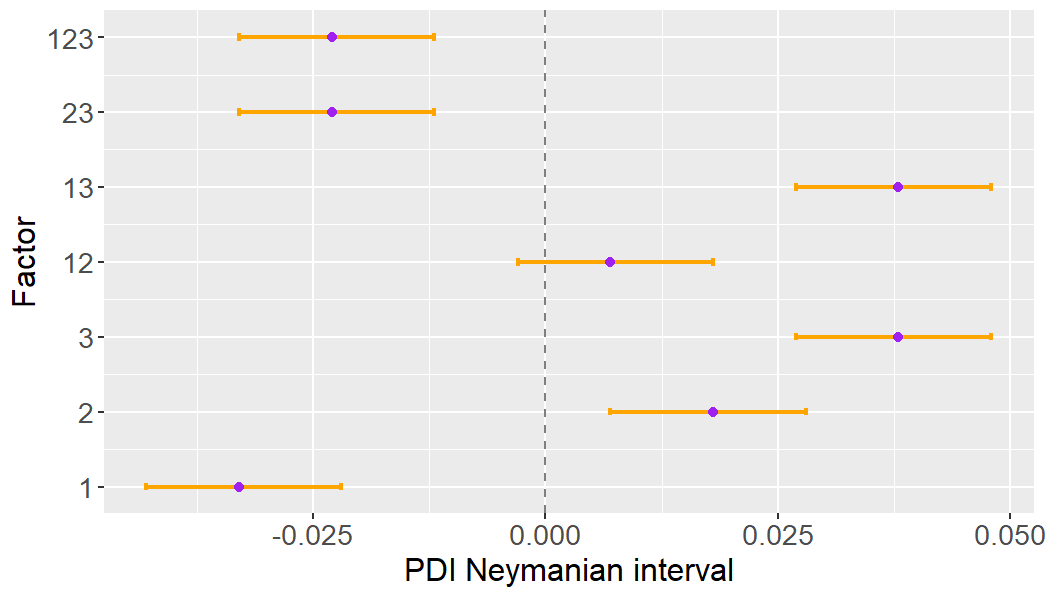}
}
\end{tabular}
\caption{95\% Neymanian confidence intervals of $2^3$ design for Example \ref{exam3}.} 
\label{fig:Example 2}
\end{figure}

\begin{table}[!http]
\caption{Summary of analysis of 95\% Neymainan intervals of $2^3$ design reported in Example \ref{exam3}.}
\footnotesize
\setlength{\tabcolsep}{0.1pt}
\label{exam32^3_sum}
\begin{tabular*}{\linewidth}{@{\extracolsep{\fill}}{c}crcrcr@{}}
\toprule
\multirow{2}{*}{Factor} 
& \multicolumn{2}{c}{MPS} 
& \multicolumn{2}{c}{EE} 
& \multicolumn{2}{c}{PDI}
\\
\cmidrule(lr){2-3}
\cmidrule(lr){4-5}
\cmidrule(lr){6-7}
&\makecell{Point\\estimate} & \makecell{Neymaian\\interval} 
& \makecell{Point\\estimate} & \makecell{Neymaian\\interval} 
& \makecell{Point\\estimate} & \makecell{Neymaian\\interval} \\
\hline
$1$ & -186.96&[-189.4,-184.6] & 9.975 & [8.154, 11.77] & -0.033 & [-0.043, -0.022] \\
$2$ & 3.125& [0.706, 5.544]       & -2.575 & [-4.396, -0.754] & 0.018 & [0.007, 0.028] \\
$3$ & 31.075   & [28.66,33.49]     & 5.875  & [4.054, 7.696]   & 0.038 & [0.027, 0.048] \\
$12$ & -9.775   & [-12.19, -7.356]    & 0.275  & [-1.546, 2.096]  & 0.007 & [-0.003, 0.018] \\
$13$ & -1.425   & [-3.844, 0.994]      & 1.925  & [0.104, 3.746]   & 0.038 & [0.027, 0.048] \\
$23$ & -31.325  & [-33.74, -28.91]   & 0.375  & [-1.446, 2.196]  & -0.023 & [-0.033, -0.012] \\
$123$& -33.925  & [-36.34, -31.51]   & -0.475 & [-2.296, 1.346]  & -0.023 & [-0.033, -0.012] \\
\hline
\end{tabular*}%
\end{table}

Main factors and their interactions $1$, $2$, $3$, $12$, $23$, and $123$ for MPS, factors $1$, $2$, $3$, and $13$ for EE, and factors $1$, $2$, $3$, $13$, $23$, and $123$ for PDI had confidence intervals indicate average causal effects. Furthermore, the main factors and their interactions reported in the original article were included among all average causal effects identified in this study.

\begin{table}[h]
\centering
\caption{Raw data for Example \ref{exam4}.}
\label{exam4_org}
\footnotesize
\setlength{\tabcolsep}{3pt}
\renewcommand{\arraystretch}{0.95}
\begin{tabular*}{\textwidth}{@{\extracolsep{\fill}}cc*{10}{r}rr@{}}
\toprule
\multicolumn{2}{c}{unit}
& \multicolumn{10}{c}{Factor}
& \multicolumn{2}{c}{Current} \\
\cmidrule(lr){1-2}
\cmidrule(lr){3-12}
\cmidrule(lr){13-14}
$i$ &$i'$& 1 & 2 & 3 & 4 & 5 & 6 & 7 & 8 & 9 & 10
& $Y_i^{obs}$ & $Y_{i'}^{obs}$ \\
\midrule
1  & 13 &  1 & -1 &  1 &  1 &  1 & -1 & -1 & -1 &  1 & -1 & 0.676 & 0.727 \\
2  & 14 &  1 &  1 &  1 & -1 & -1 & -1 &  1 & -1 &  1 &  1 & 2.913 & 4.020 \\
3  & 15 & -1 &  1 &  1 &  1 & -1 & -1 & -1 &  1 & -1 &  1 & 4.121 & 4.649 \\
4  & 16 &  1 & -1 & -1 & -1 &  1 & -1 &  1 &  1 & -1 &  1 & 0.678 & 0.841 \\
5  & 17 &  1 &  1 & -1 & -1 & -1 &  1 & -1 &  1 &  1 & -1 & 1.679 & 1.893 \\
6  & 18 & -1 & -1 & -1 &  1 & -1 &  1 &  1 & -1 &  1 &  1 & 2.039 & 2.210 \\
7  & 19 &  1 & -1 &  1 &  1 & -1 &  1 &  1 &  1 & -1 & -1 & 2.040 & 3.141 \\
8  & 20 & -1 &  1 &  1 & -1 &  1 &  1 &  1 & -1 & -1 & -1 & 1.369 & 1.898 \\
9  & 21 & -1 & -1 & -1 & -1 & -1 & -1 & -1 & -1 & -1 & -1 & 3.008 & 3.821 \\
10 & 22 & -1 & -1 &  1 & -1 &  1 &  1 & -1 &  1 &  1 &  1 & 0.715 & 0.839 \\
11 & 23 & -1 &  1 & -1 &  1 &  1 & -1 &  1 &  1 &  1 & -1 & 5.533 & 5.981 \\
12 & 24 &  1 &  1 & -1 &  1 &  1 &  1 & -1 & -1 & -1 &  1 & 4.870 & 5.622 \\
\bottomrule
\end{tabular*}
\end{table}

{\bf Example \ref{exam4}.} (Cont.) \cite{hardianto2023identification} provided two replicate observed outcomes, as shown in Table \ref{exam4_org}. The PB design with 12 runs is incomplete and non-geometric. 
Units in the first and second replicate trials are indexed by $i=1,\ldots,12$ and $i'=13,\ldots,24$, respectively. The corresponding observed outcomes are denoted by $Y_i^{obs}$ and $Y_{i'}^{obs}$. The point estimate $\hat\tau(k)$ calculated using (\ref{Yobs}), and the confidence interval for $\hat\tau(k)$ under Neymanian inference obtained based on (\ref{cor}) and (\ref{NeyCI}). Because there are only 10 factors, the original design matrix does not satisfy the complete PB design structure required by Algorithm \ref{alg1}. To apply the algorithm under Fisherian inference, we add a dummy factor column, denoted by $g_{11}$, and denote the resulting complete design matrix by ${\bf G}$. The dummy factor $11$ is introduced solely to complete the design and is assigned an associated effect of 0. Algorithm \ref{alg1} is then applied to ${\bf G}$ to construct the confidence interval for $\hat\tau(k)$ under Fisherian inference. The point estimates and confidence intervals for Neymanian and Fisherian inferences are summarized in Table \ref{exam4_sum}, and the plots are shown in Figure \ref{exam4_figure}.

\begin{figure}[h]
\centering
\begin{tabular}{cc}
\includegraphics[width=0.6\textwidth]{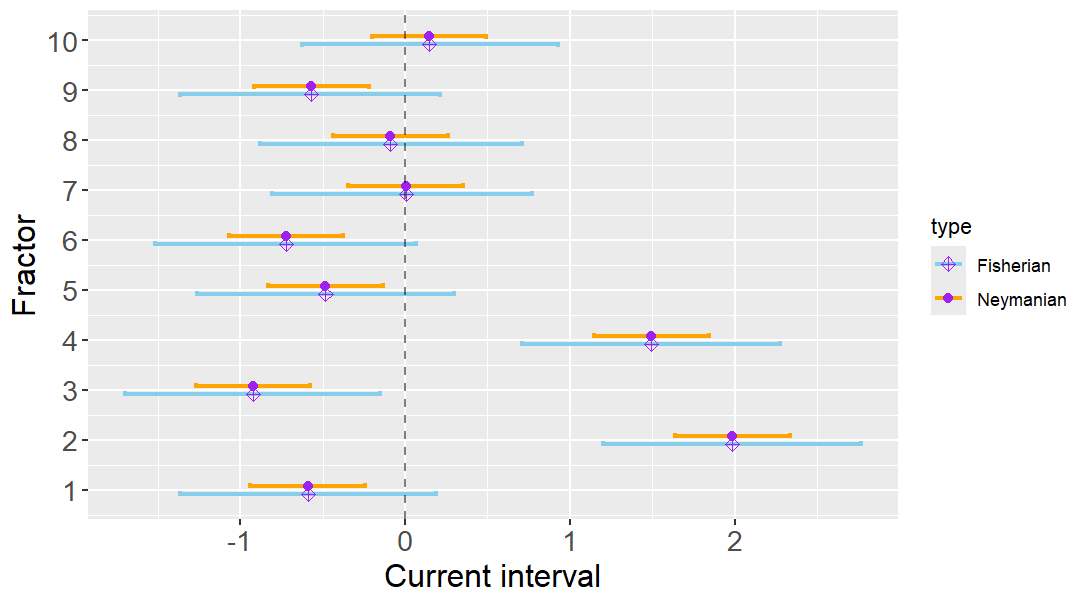} 
\end{tabular}
\caption{95\% Neymanian confidence intervals for Example \ref{exam4}.} 
\label{exam4_figure}
\end{figure}

\begin{table}[h]
\caption{Summary of analysis of 95\% Neymainan intervals of PB design reported in Example \ref{exam4}.}
\footnotesize
\setlength{\tabcolsep}{0.1pt}
\label{exam4_sum}
\begin{tabular*}{\linewidth}{@{\extracolsep{\fill}}{c}crc@{}}
\hline
\makecell{Factor}
& \makecell{Point estimate} & \makecell{Neymaian interval} 
& \makecell{Fisherian interval} \\
\hline
 1  & -0.590 & $[-0.938,\,-0.242]$ & $[-1.367,\,0.190]$ \\
2  &  1.984 & $[1.636,\,2.333]$   & $[1.202,\,2.768]$ \\
3  & -0.922 & $[-1.270,\,-0.574]$ & $[-1.700,\,-0.151]$ \\
4  &  1.495 & $[1.146,\,1.843]$   & $[0.710,\,2.277]$ \\
5  & -0.482 & $[-0.830,\,-0.134]$ & $[-1.263,\,0.295]$ \\
6  & -0.721 & $[-1.069,\,-0.373]$ & $[-1.516,\,0.070]$ \\
7  &  0.004 & $[-0.345,\,0.352]$  & $[-0.804,\,0.769]$ \\
8  & -0.089 & $[-0.437,\,0.260]$  & $[-0.880,\,0.711]$ \\
9  & -0.569 & $[-0.918,\,-0.221]$ & $[-1.366,\,0.210]$ \\
10 &  0.146 & $[-0.202,\,0.494]$  & $[-0.624,\,0.927]$ \\
\hline
\end{tabular*}%
\end{table}

In Figure \ref{exam4_figure} and Table \ref{exam4_sum}, for the outcome variable current, Neymanian inference indicates that the average causal effects of factors 1-6 and 9 are all nonzero, whereas Fisherian inference indicates that only the average causal effects of factors 2-4 are nonzero. In addition, the confidence intervals obtained by Fisherian inference are wider than those obtained by Neymanian inference, which is consistent with the result reflected in (\ref{balanced}), namely that the variance under Fisherian inference is larger than that under Neymanian inference. Therefore, under the same confidence level, the confidence intervals produced by Fisherian inference are more likely to contain 0. The original study, using the ANOVA method, identified factors 1-6 and 9, which is consistent with the result obtained by Neymanian inference. By comparison, the result obtained by Fisherian inference is more conservative.

\section{Discussion}\label{sec6}
Most applications of PB design have been studied from the perspective of statistical inference, whereas this paper analyzes them from the perspective of causal inference. Based on the potential outcome, we defined the causal effect of the PB design and performed the Neymanian inference and the Fisherian inference. In Neymanian inference, we derive variance and construct a Neymanian interval. In Fisherian inference, we provide pseudocode for calculating the Fisherian intervals of the PB design. We applied the method to four examples and interpreted the role of each factor based on the defined causal effect, thereby improving the interpretability of the results. Furthermore, the factors identified by the average causal effects across the four examples fully covered all significant factors identified in the original paper through statistical inference, with no omissions.

Several future directions stem from this paper. First, we examine PB design in causal inference from a design perspective, though it could also be approached through modeling. Second, we construct Bayesian inference for PB design and compare it to other inferences.
\section*{Acknowledgments} The work was supported by the National Natural Science Foundation of China (12561047), the Xinjiang Talent Development Fund (XJRC-2025-KJ-PY-KJLJ-108), and the 2025 Central Guidance for Local Science and Technology Development Fund (ZYYD2025ZY20).
\section*{Conflict of interest} 
The authors declare that they have no conflict of interest.


\newpage
\section*{Appendix: Algorithm 1}
\begin{algorithm}[H]
\caption{Fisherian confidence interval for the $k$th factorial effect $\tau(k)$.}\label{alg1}
\KwData{Assignment variable $W$, Significance level $\alpha$, Permutations number $B$, $g_k(k=1,\ldots,g_{N-1})$, Observed outcomes $Y_i^{obs}(i=1,\ldots,n)$, Initial interval $[{\hat\tau(k,W,Y^{obs})-c, \hat\tau(k,W,Y^{obs})+c}]$, vector ${\bm \eta}=(\eta_1,\ldots, \eta_{N-1})$.}
\KwResult{Fisherian interval $[L_k,U_k]$.}
\Fn{\Pv{$\eta^{*}_k$}}{
${\tilde{\bm\eta}}\gets{\bm \eta}$,\quad $\tilde {\bm\eta}[k]\gets{\eta^{*}_k}$\;
\For{$i = 1 \to n$}
{
$\mu_i=Y^{obs}_i-z^{obs}_{i}{ \tilde {\bm\eta}}/2$\;
{\For{$j = 1 \to N$}
{$Y_i({\bf z}_j)=\mu_i+z_{j}{ \tilde {\bm\eta}}/2$\;}}}
$count \gets 0$\;
\For{$b=1\to B$}
{\(W^b\) is a permutation of \(W\)\;
\For{$j = 1 \to N$}
{$\bar Y^{b,obs}({\bf z}_j)=\frac{1}{n_j}W_i^b({\bf z}_j)Y_i({\bf z}_j)$\;}
$\hat\tau(k,W^b,\bar Y^{b,obs})=\frac{2}{N}g_k\bar Y^{b,obs}$\;
{\If {$\hat\tau(k,W^b,\bar Y^{b,obs})\geq \hat\tau(k,W,\bar Y^{obs})$}{
 $count \gets count + 1$\;}}}
 $p\gets count/B$\;
\Return $p$}
\Fn{\Bo{$u$}}{
$a \gets \hat\tau(k,W,\bar Y^{obs})-c$, $b \gets \hat\tau(k,W,\bar Y^{obs})+c$, $p(a) \leftarrow pvalue(a)$, $p(b) \leftarrow pvalue(b)$\;
\For{$iter<iter_{max}$}
{$m=(a+b)/2$, $p(m) \leftarrow pvalue(m)$\;}
\eIf{$(p(a)-u)(p(m)-u)\leq0$}{
    $b \gets m$\;
}{$a \gets m$, $p(a) \gets p(m)$\;}
}
$L_k \leftarrow bound(\alpha/2)$\;
$U_k \leftarrow bound(1-\alpha/2)$\;
\Return $[L_k,U_k]$
\end{algorithm}

\end{document}